\documentclass[prx,twocolumn,floatfix,superscriptaddress,longbibliography,notitlepage]{revtex4-1}

\usepackage{graphicx, color, graphpap}
\usepackage{enumitem}
\usepackage{amssymb}
\usepackage{amsthm}
\usepackage{multirow}
\usepackage[colorlinks=true,citecolor=blue,linkcolor=magenta]{hyperref}
\usepackage[T1]{fontenc}
\usepackage{bbm}
\usepackage{thmtools,thm-restate}
\usepackage{verbatim}
\usepackage{mathtools}
\usepackage{titlesec}
\usepackage{amsmath}

\long\def\ca#1\cb{} 

\newcommand{\ketbra}[2]{| \hspace{1pt} #1 \rangle \langle #2 \hspace{1pt} |}

\newcommand{\ket}[1]{|#1\rangle}               
\newcommand{\bra}[1]{\langle #1|}              
\newcommand{\dya}[1]{\ket{#1}\!\bra{#1}}

\newcommand{\poly}{\operatorname{poly}}

\newcommand{\HC}{\mathcal{H}}

\newcommand{\OC}{\mathcal{O}}

\newcommand{\SC}{\mathcal{S}}

\newcommand{\Tr}{{\rm Tr}}

\newcommand{\Var}{{\rm Var}}

\renewcommand{\geq}{\geqslant}
\renewcommand{\leq}{\leqslant}

\renewcommand{\vec}[1]{\boldsymbol{#1}}  

\newcommand{\ad}{^\dagger}

\newcommand*{\id}{\openone}

\newcommand{\rhoi}{\rho^{\text{in}}}
\newcommand{\rhoo}{\rho^{\text{out}}}

\newcommand{\tout}{{\text{out}}}
\newcommand{\tin}{{\text{in}}}
\newcommand{\thid}{{\text{hid}}}

\newcommand{\thv}{\vec{\theta}}

{}
{}
\newtheorem{theorem}{Theorem}

\newtheorem{lemma}{Lemma}

\begin{document}
\title{Trainability of Dissipative Perceptron-Based  Quantum Neural Networks}

\author{Kunal Sharma} 
\thanks{The first two authors contributed equally to this work.}
\address{Theoretical Division, Los Alamos National Laboratory, Los Alamos, New Mexico 87545, USA}
\address{Hearne Institute for Theoretical Physics and Department of Physics and Astronomy, Louisiana State University, Baton Rouge, Louisiana 70803, USA}

\author{M. Cerezo}
\thanks{The first two authors contributed equally to this work.}
\affiliation{Theoretical Division, Los Alamos National Laboratory, Los Alamos, New Mexico 87545, USA}
\affiliation{Center for Nonlinear Studies, Los Alamos National Laboratory, Los Alamos, New Mexico 87545, USA
}

\author{Lukasz Cincio}
\affiliation{Theoretical Division, Los Alamos National Laboratory, Los Alamos, New Mexico 87545, USA}

\author{Patrick J. Coles}
\affiliation{Theoretical Division, Los Alamos National Laboratory, Los Alamos, New Mexico 87545, USA}

\begin{abstract}
Several architectures have been proposed for quantum neural networks (QNNs), with the goal of efficiently performing machine learning tasks on quantum data.  Rigorous scaling results are urgently needed for specific QNN constructions to understand which, if any, will be trainable at a large scale. Here, we analyze the gradient scaling (and hence the trainability) for a recently proposed architecture that we call dissipative QNNs (DQNNs), where the input qubits of each layer are discarded at the layer's output. We find that DQNNs can exhibit barren plateaus, i.e., gradients that vanish exponentially in the number of qubits. Moreover, we provide quantitative bounds on the scaling of the gradient for DQNNs under different conditions, such as different cost functions and circuit depths, and show that trainability is not always guaranteed. Our work represents the first rigorous analysis of the scalability of a perceptron-based QNN.
\end{abstract}
\maketitle

{\it Introduction.---}Neural networks (NN) have impacted many fields such as neuroscience, engineering, computer science, chemistry, and physics~\cite{haykin1994neural}. However, their historical development has seen periods of great progress interleaved with periods of stagnation, due to serious technical challenges~\cite{minsky2017perceptrons}. The perceptron was introduced early on as an artificial neuron~\cite{rosenblatt1957perceptron}, but it was only realized later that a multilayer perceptron (now known as a feedforward NN) had much greater power than a single-layer one~\cite{haykin1994neural,minsky2017perceptrons}. Still there was the major issue of how to train multiple layers, and this was eventually addressed by the backpropagation method~\cite{rumelhart1986learning}.

Motivated by the success of NNs and the advent of  noisy intermediate-scale quantum devices~\cite{preskill2018quantum}, there has been tremendous effort to develop quantum neural networks (QNNs)~\cite{schuld2014quest}. The hope is that QNNs will harness the power of quantum computers to outperform their classical counterparts on machine learning tasks~\cite{nielsen2015neural,biamonte2017quantum}, especially for quantum data or tasks that are inherently quantum in nature~\cite{sharma2020reformulation}.

Despite several QNN proposals that have been successfully implemented~\cite{Romero,dunjko2018machine,verdon2018universal,farhi2018classification,ciliberto2018quantum,killoran2019continuous,cong2019quantum,jia2019quantum}, more research is needed on the advantages and limitations of specific architectures. Delving into potential scalability issues of QNNs could help to prevent a ``winter'' for these models, like what was seen historically for classical NNs. This has motivated recent works studying the scaling of gradients in QNNs~\cite{mcclean2018barren,cerezo2020cost}. There, it was shown that variational quantum algorithms~\cite{VQE,bauer2016hybrid,mcclean2016theory,arrasmith2019variational,jones2019variational,Xiaosi,bravo-prieto2019,yuan2019theory,cirstoiu2019variational,cerezo2019variational,cerezo2020variational}, which aim to train QNNs to accomplish specific tasks, may exhibit gradients that vanish exponentially with the system size. This so-called barren-plateau phenomenon, where the parameters cannot be efficiently trained for large implementations, was demonstrated for hardware-efficient QNNs, where quantum gates are arranged in a bricklike structure that 
matches the connectivity of the quantum device~\cite{mcclean2018barren,cerezo2020cost}. 

Analyzing the existence of barren plateaus in QNNs is paramount to determining if they can lead to a quantum  speedup. This is due to the fact that exponentially vanishing gradients imply that the precision needed to estimate such gradients grows exponentially. Since the standard goal of quantum algorithms is polynomial scaling as opposed to the typical exponential scaling of classical algorithms, a QNN with exponentially vanishing gradients has no hope of achieving this goal. On the other hand, a QNN with gradients that vanish polynomially means that the algorithm requires a polynomial precision, and hence that the hope of quantum speedup is preserved.

Here, we analyze the trainability and the existence of barren plateaus in a class of QNNs that we refer to as {\it dissipative QNNs} (DQNNs). In a DQNN each node within the network corresponds to a qubit~\cite{kouda2005qubit}, and the connections in the network are modelled by quantum perceptrons~\cite{altaisky2001quantum,sagheer2013autonomous,siomau2014quantum,torrontegui2019unitary,tacchino2019artificial,beer2020training}. The term dissipative refers to the fact that ancillary qubits form the output layer, while the qubits from the input layer are discarded. This architecture has seen significant recent attention and has been proposed as a scalable approach to QNNs~\cite{beer2020training,bondarenko2020quantum,poland2020no}. In particular, in Ref.~\cite{beer2020training}, based on small scale numerical experiments, it was speculated that dissipative quantum neural networks do not suffer from the barren plateau (vanishing gradient) problem. However, contrary to Ref.~\cite{beer2020training}, we here analytically prove that DQNNs are not immune to barren plateaus. For example, DQNNs with deep global perceptrons are untrainable despite the dissipative nature of the architecture.

Here we study the large-scale trainability of DQNNs. In particular, we focus on tasks where DQNNs are employed to learn a unitary matrix connecting input and output quantum states and for general supervised quantum machine learning tasks where training data consists of quantum states and corresponding classical labels. For these tasks, we show that the barren plateau phenomenon can also arise in DQNNs. We also discuss certain  conditions (e.g., the structure and depth of the DQNN) under which one could avoid a barren plateau and achieve trainability. In particular, our work implies that scalability is not guaranteed, and without careful thought of the  structure of DQNNs, their gradients may vanish exponentially in the system size. As a by-product of our analysis of specific perceptron architectures, we also show that hardware-efficient QNNs are special cases of DQNNs. Therefore, many important results for hardware-efficient QNNs, such as the ones studied in Refs.~\cite{mcclean2018barren,cerezo2020cost} also hold for DQNNs. Finally, we remark that we employ novel analytical techniques in our proofs (different from those used in Refs.~\cite{mcclean2018barren,cerezo2020cost}), which were necessary to develop due to the dissipative nature of DQNNs. Our techniques may be broadly useful in the study of the scaling of other QNN architectures. 

\bigskip

{\it Preliminaries.---} Let us first introduce the DQNN architecture. As schematically shown in Fig.~\ref{fig:f1}, the DQNN is composed of a series of layers (input, hidden, and ouput) where the qubits at each node are connected via perceptrons. A quantum perceptron is defined  as an arbitrary unitary operator with $m$ input and $k$ output qubits. For simplicity, we consider the case when $k=1$, so that each perceptron acts on $m+1$ qubits. The case of arbitrary $k$ is presented in the Supplemental Material.

\begin{figure}[t]
    \centering
    \includegraphics[width=.8\columnwidth]{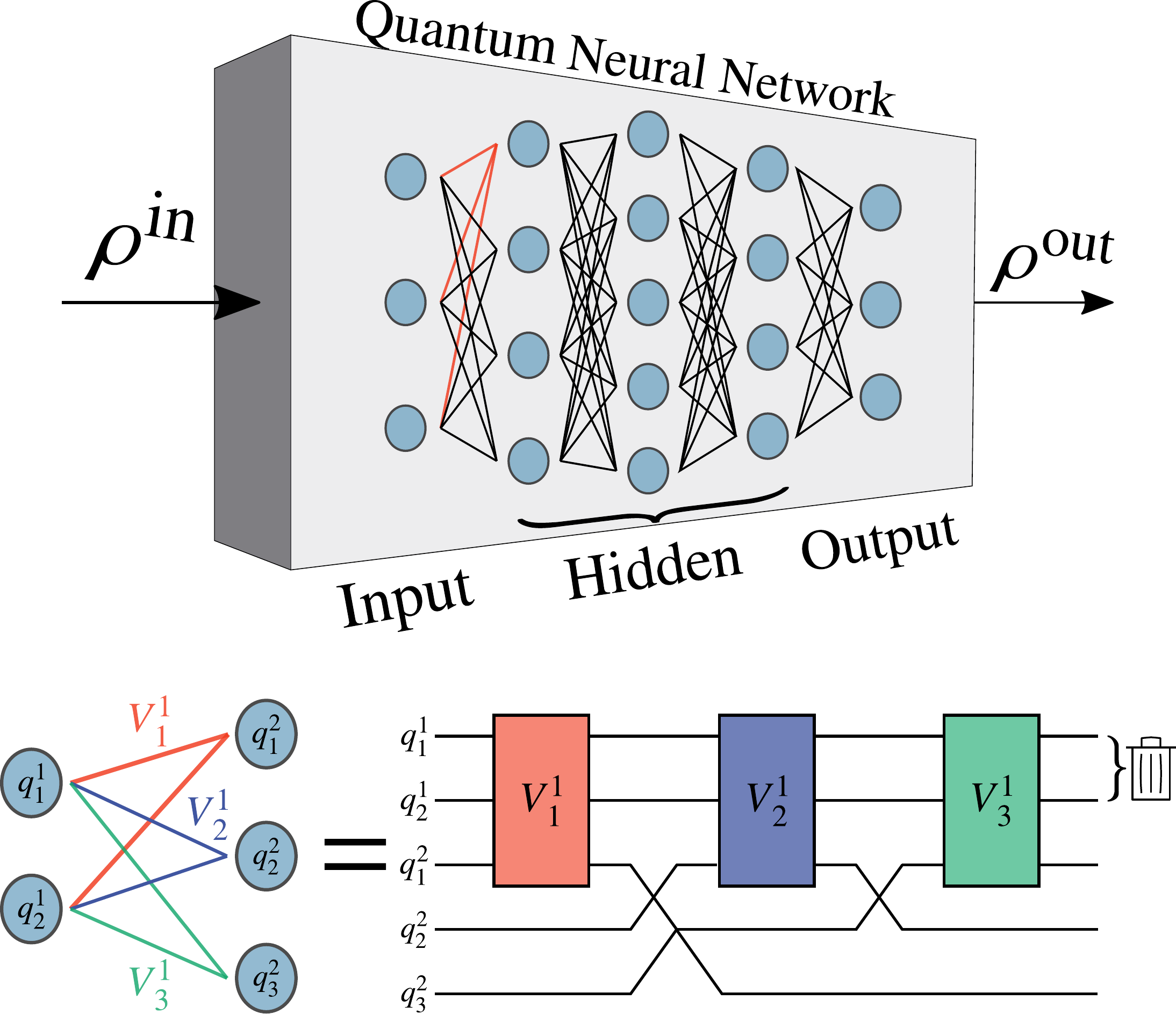}
    \caption{Schematic diagram of a dissipative perceptron-based quantum neural network (DQNN). Top: The DQNN is composed of input, hidden, and output layers. Each node in the network corresponds to a qubit, which can be connected to qubits in adjacent layers via perceptrons (depicted as lines). The input and output of the DQNN are quantum states denoted as $\rho^\tin$ and $\rho^\tout$, respectively. Bottom: Quantum circuit description of the DQNN. The $j$th qubit of the $l$th layer is denoted $q_j^l$. Each perceptron corresponds to a unitary operation on the qubits it connects, with $V_j^l$ denoting the $j$-th perceptron in the $l$-th layer.}
    \label{fig:f1}
\end{figure}

The qubits in the input layer are initialized to a state $\rhoi$, while all qubits in the hidden and output layers are initialized to a fiduciary state such as $\ket{\vec{0}}_{\thid,\tout}=\ket{0\ldots0}_{\thid,\tout}$. Henceforth we employ the  notation ``in'', ``hid'', and ``out'' to indicate operators on qubits in the input, hidden, and output layers, respectively. 
The output state of the DQNN is a quantum state $\rhoo$ (generally mixed) which can be expressed as  
\begin{equation}\label{eq:rhoutgeneral}
    \rhoo\equiv \Tr_{\text{in},\text{hid}}\left[V(\rhoi\otimes\ket{\vec{0}}_{\text{hid},\text{out}}\bra{\vec{0}})V\ad\right]\,,
\end{equation}
with $V=V^{\text{out}}_{n_{\text{out}}}\ldots V^1_{n_1} \ldots V^1_1$, and where $V^l_j$ is the perceptron unitary on the $l$-th layer acting on the $j$-th output qubit. Here  $n_l$ indicates the number of qubits in the $l$-th layer. 

Let us now make two important remarks. First, note that the order in which the perceptrons act is relevant, as in general the unitaries $V_j^l$ will not commute. Second, we remark that for this architecture the perceptrons are applied layer by layer, meaning that once all $V_j^l$ (for fixed $l$) have been applied and the information has propagated forward between layers $l-1$ and $l$, one can discard the qubits in layer $l-1$. This implies that the width of the DQNN depends on the number of qubits in two adjacent layers and not in the total number of qubits in the network.

To train the DQNN, we assume repeatable access to training data in the form of pairs $\{\ket{\phi_x^{\tin}},\ket{\phi_x^{\tout}}\}$, with $x=1,\ldots,N$.  We note that, as discussed in the Supplemental Material, our results also hold more generally for supervised quantum machine learning tasks where the training data is of the form $\{\ket{\phi_x^{\tin}},y_x\}$, with $y_x$ a label assigned to the input state $\ket{\phi_x^{\tin}}$~\cite{havlivcek2019supervised}.

We then define a cost function (or loss function) which quantifies  how well the DQNN reproduces the training data. We assume that the cost is of the form 
\begin{equation}\label{eq:costx}
    C=\frac{1}{N}\sum_{x=1}^N C_x\,, \quad \text{with} \quad C_x =\Tr[O_x \rho^\tout_x]\,.
\end{equation}
As discussed below, in general there are multiple choices for the operator $O_x$ which lead to faithful cost functions, i.e., costs that are extremized if and only if one perfectly learns the mapping on the training data. If the circuit description of output states is provided, one can employ the inverse of the corresponding unitary on the output of a DQNN \cite{sharma2019noise}. Then a measurement in the computational basis estimates the cost function. Otherwise, one can employ a recently developed procedure based on classical shadows to estimate the state overlap \cite{huang2020predicting}.

When $O_x$ acts non-trivially on all qubits of the output layer, we use the term {\it global cost function}, denoted as $C^G$. Here one usually compares objects (states or operators) living in exponentially large Hilbert spaces. For instance, choosing 
\begin{equation}\label{eq:globalop}
    O_x^G=\id - \ketbra{\phi_x^{\tout}}{\phi_x^{\tout}}\,,
\end{equation}
leads to a global cost function that quantifies the average fidelity between each $\rho^\tout_x$ and  $\ket{\phi_x^{\tout}}$.

As shown in Ref.~\cite{cerezo2020cost}, local cost functions do not exhibit a barren plateau for shallow hardware-efficient QNNs. Therefore, it is important to study if local observables can also lead to trainability guarantees  in DQNNs. Henceforth, we use the term {\it local cost function}, denoted as $C^L$, for the cases when the operator $O_x$ acts non-trivially on a small number of qubits in the output layer. 
Since the global cost in Eq.~\eqref{eq:globalop} is a state fidelity function, in general it will not be possible to design a corresponding faithful local cost. Therefore, we restrict ourselves to the case when $\ket{\phi_x^{\tout}}$ is a tensor-product state across $n_\tout$ qubits $\ket{\phi_x^{\tout}}=\ket{\psi^{\tout}_{x,1}}\otimes\ldots \otimes \ket{\psi^{\tout}_{x,n_\tout}}$. Then, we can define the  following local observable:
\begin{equation}\label{eq:localop}
    O_x^L=\id - \frac{1}{n_\tout}\sum_{j=1}^{n_\tout}\dya{\psi^{\tout}_{x,j}}\otimes \id_{\overline{j}}\,,
\end{equation}
where $\id_{\overline{j}}$ denotes the identity over all qubits in the output layer except for qubit $j$.  Equation~\eqref{eq:localop} leads to a faithful local cost that vanishes under the same condition as the global cost defined from Eq.~\eqref{eq:globalop}~\cite{QAQC,sharma2019noise}. 

\begin{figure}[t]
    \centering
    \includegraphics[width=.75\columnwidth]{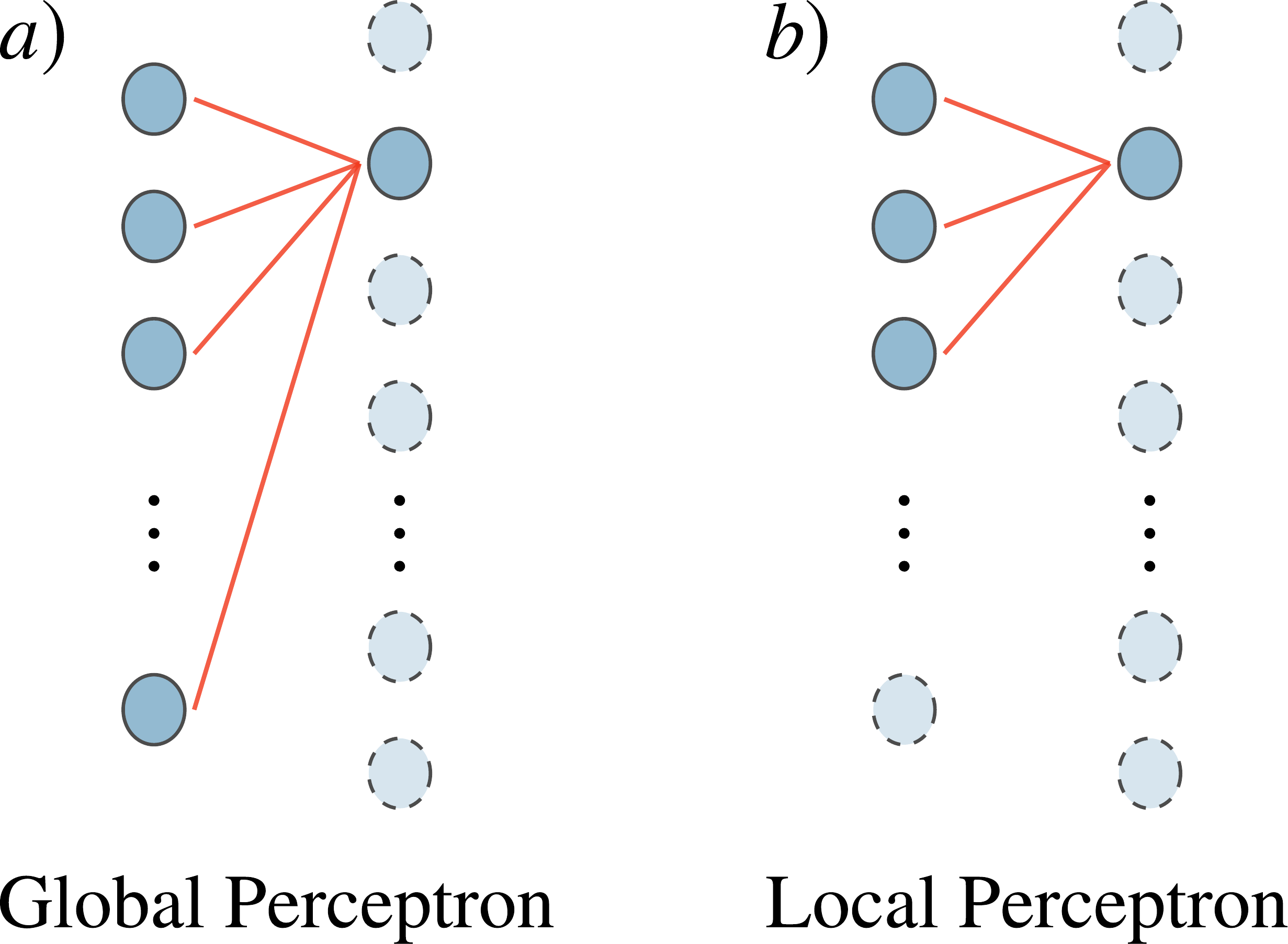}
    \caption{Global and local perceptrons. a) The global perceptron acts non-trivially on all input qubits, i.e.,  $m=n$. b) The local perceptron acts non-trivially only on a small number of input qubits. For the case shown, $m=3$. }
    \label{fig:f2}
\end{figure}

Finally let us introduce the term {\it global perceptron} to refer to the case when the perceptron $V_j^l$ acts non-trivially on {\it all}  qubits in the $l$-th layer, i.e., when $m=n_{l-1}$. On the other hand, a {\it local perceptron} is defined as a unitary  $V_j^l$  acting on a number of qubits $m\in\OC(1)$ which is independent of $n_{l-1}$. Figure~\ref{fig:f2} schematically shows a global and a local perceptron.

To analyze the existence of barren plateaus and the trainability of the DQNN one needs to define an ansatz and a training method for the perceptrons. In what follows we consider two general training approaches.

\bigskip

{\it Random parameterized quantum circuits.---}We first consider the case where the perceptrons are parametrized quantum circuits (i.e., variational circuits) that can be expressed as a sequence of parameterized and unparameterized gates from a given gate alphabet~\cite{mcclean2018barren,du2018expressive}. That is, the perceptrons are of the form
\begin{equation}\label{eq:parametrization}
V_j^l(\thv_j^l)=\prod_{k=1}^{\eta_j^l}R_k(\theta^k)W_k\,,
\end{equation}
with $R_k(\theta^k)=e^{-(i/2) \theta^k \Gamma_k}$,  $W_k$ an unparameterized unitary,  and where $\Gamma_k$ is a Hermitian operator with $\Tr[ \Gamma_k^2]\leq 2^{n+1}$. 
Such parameterization is widely used as it can allow for a straightforward evaluation of the cost function gradients, and since in general its quantum circuit description can be easily obtained~\cite{guerreschi2017practical,mitarai2018quantum,schuld2019evaluating}. 

A common strategy for training random parameterized quantum circuits is to randomly initialize the parameters in~\eqref{eq:parametrization}, and employ a training loop to minimize the cost function. To analyze the trainability of the DQNN we compute the variance of the partial derivative $\partial C/\partial \theta^\nu\equiv \partial_\nu C$, where $\theta^\nu$ belongs to a given  $V_j^l$
\begin{equation}
    \Var[\partial_\nu C]=\left\langle(\partial_\nu C)^2\right\rangle-\left\langle\partial_\nu C\right\rangle^2\,.
\end{equation}
Here the notation $\langle\cdots\rangle$ indicates the average over all randomly initialized perceptrons. 
From~\eqref{eq:parametrization}, we find
\small
\begin{align}\label{eq:parC}
\partial_\nu C= \frac{i}{2N}\sum_{x=1}^{N}&\Tr\Big[ A_j^l\tilde{\rho}^\tin_x (A_j^l)\ad[\id_{\overline{j}}^{\overline{l}}\otimes\Gamma_k,(B_j^l)\ad\tilde{O}_x B_j^l]\Big]\,,
\end{align}
\normalsize
where we have defined
\begin{equation}\label{eq:AandB}
    B_j^l=\id_{\overline{j}}^{\overline{l}}\otimes\prod_{k=1}^{\nu-1}\!\!R_k(\theta^k)W_k\,,\,\,
    A_j^l=\id_{\overline{j}}^{\overline{l}}\otimes\prod_{k=\nu}^{\eta_j^l}\!\!R_k(\theta^k)W_k\,,
\end{equation}
such that $\id_{\overline{j}}^{\overline{l}}\otimes V_j^l=A_j^lB_j^l$, and where $\id_{\overline{j}}^{\overline{l}}$ indicates the identity on all qubits on which $V_j^l$ does not act. Note that the trace in \eqref{eq:parC} is over {\it all} qubits in the DQNN.  In addition, we define
\begin{align}
    \tilde{\rho}^\tin_x&=V_{j-1}^l\ldots V_1^1 (\rho^\tin_x\otimes\dya{\vec{0}}_{\thid,\tout}) (V_1^1)\ad\ldots (V_{j-1}^l)\ad\,,\nonumber\\
    \tilde{O}_x&=(V_{j+1}^l)\ad\ldots (V_{n_\tout}^\tout)\ad (\id_{\tin,\thid}\otimes O_x) V_{n_\tout}^\tout\ldots V_{j+1}^l\,.\nonumber
\end{align}

If the perceptron $V_j^l$ is sufficiently random so that $A_j^l$, $B_j^l$, or both, form independent unitary $1$-designs, then we find that $\langle \partial_\nu C \rangle=0$ (see Supplemental Material). In this case,  $\Var[\partial_s C]$ quantifies (on average) how much the gradient concentrates around zero. Hence, exponentially small $\Var[\partial_s C]$ values would imply that the slope of the cost function landscape is insufficient to provide cost-minimizing directions.

Here we recall that a $t$-design is a set of unitaries $\{V_y\in U(d)\}_{y\in Y}$ (of size $|Y|$) on a $d$-dimensional Hilbert space such that for every polynomial $P_t(V_y)$ of degree at most $t$ in the matrix elements of $V_y$, and of $V_y\ad$ one has~\cite{dankert2009exact}
$
  \langle P_{t}(V)\rangle_{V}= \frac{1}{|Y|}\sum_{y\in Y}P_{t}(V_y)=\int  d\mu(V)P_t(V),
$
where the integral is over the unitary group $U(d)$. 

Let us assume for simplicity the case when the DQNN input and output layers have the same number of qubits ($n_\tin=n_\tout=n$).  As shown in the Supplemental Material, the following theorem holds.
\begin{theorem}\label{theo2}
Consider a DQNN with deep global perceptrons parametrized as in Eq.~\eqref{eq:parametrization}, such that  $A_j^l$, $B_j^l$ in Eq.~\eqref{eq:AandB} and $V_j^l$ $(\forall j,l)$ form independent $2$-designs over $n+1$ qubits. Then, the variance of the partial derivative of the  cost function with respect to $\theta^\nu$ in $V_j^l$ is upper bounded as
\begin{equation}
\Var[\partial_\nu C^G]\leq g(n), \quad \text{with} \quad g(n)\in\OC\left(1/2^{2n}\right)\,,
\end{equation}
if $O_x$ is the global operator of Eq.~\eqref{eq:globalop}, and upper bounded as 
\begin{equation}
\Var[\partial_\nu C^L]\leq h(n), \quad \text{with} \quad h(n)\in\OC\left(1/2^{n}\right)\,,
\end{equation}
when $O_x$ is the local operator in~\eqref{eq:localop}. 
\end{theorem}

Theorem~\ref{theo2} shows that DQNNs with deep global perceptron unitaries that form two-designs~\cite{brandao2016local,harrow2018approximate} exhibit barren plateaus for global and local cost functions. An immediate question that follows is whether barren plateaus still arise for shallow perceptrons, which cannot form $2$-designs on $n+1$ qubits. In what follows we analyze specific cases of shallow local perceptrons for which results can be obtained.  

\begin{figure}[t]
    \centering
    \includegraphics[width=1\columnwidth]{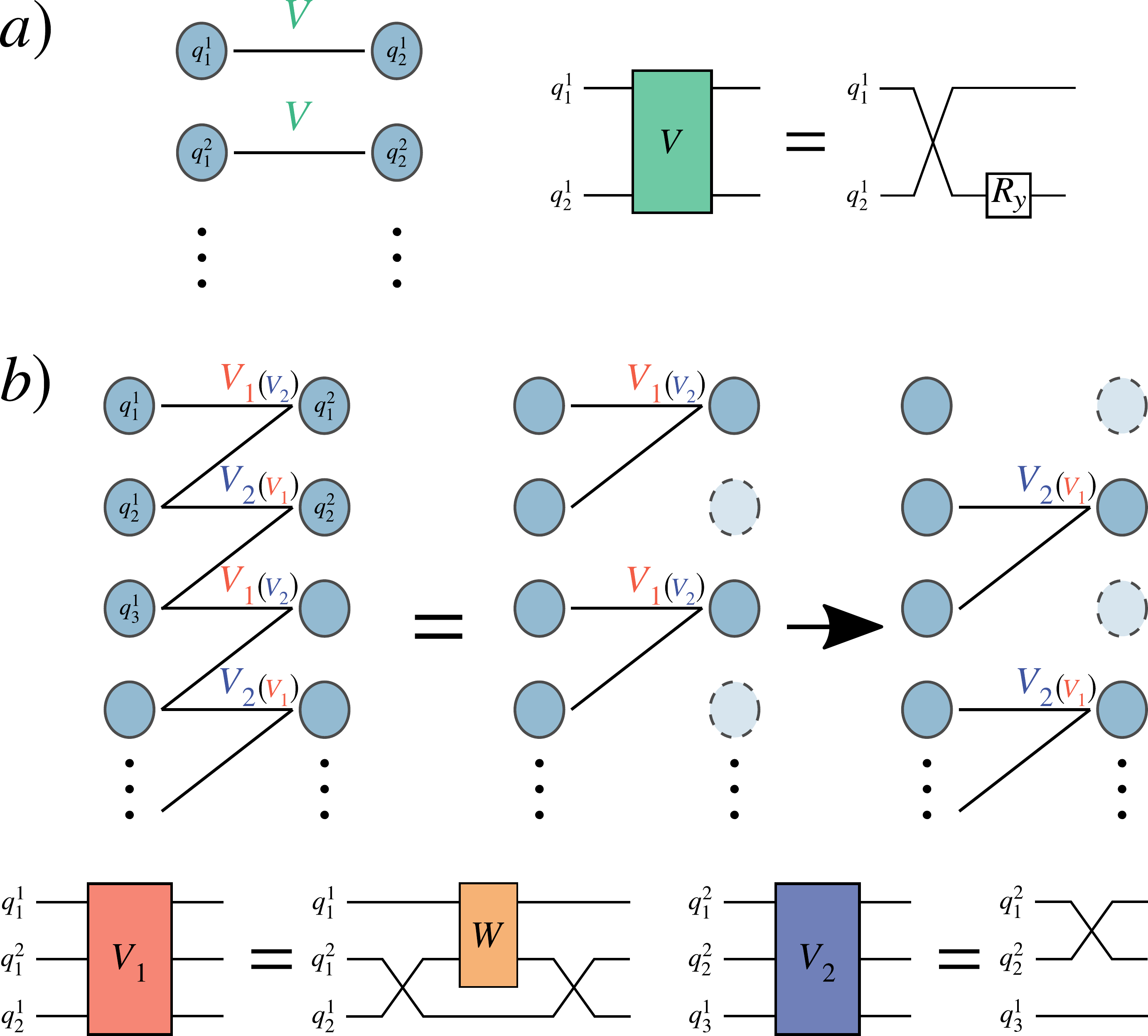}
    \caption{Shallow local perceptrons ansatzes. a) Here $m=1$ so that each perceptron acts on a single input and output qubit. Moreover, for all $j$ and $l$ we have $V_j^l=V$. The unitaries $V$ are simply given by a SWAP operator followed by a single qubit rotation around the $y$ axis. b) Local perceptrons $V_j^l$ with $m=2$. The local perceptrons are given by the unitaries $V_1$, or $V_2$. Specifically, for $l$ odd on $j$ odd (even) $V_j^l=V_1(V_2)$, while for $l$ even and $j$ odd (even) we have $V_j^l=V_2(V_1)$.   Here we also show the order in which the perceptrons are applied so that we first implement the unitaries with $j$ odd, followed by the unitaries with $j$ even.
    The $W$ gate in $V_1$ forms a local $2$-design on two qubits.      }
    \label{fig:f3}
\end{figure}

Let us first consider the simple perceptrons of Fig.~\ref{fig:f3}(a), where $m=1$, and where $R_y$ denotes a single qubit rotation around the $y$ axis: $R_y(\theta^\nu)=e ^{-i\theta^\nu Y/2}$ (with all angles randomly initialized). 
In this case one recovers the toy model example of~\cite{cerezo2020cost}, and we know that if $O_x$ is the global operator of~\eqref{eq:globalop}, then 
$\Var[\partial_\nu C^G]=\frac{1}{8}\left(\frac{3}{8}\right)^{n-1}$. On the other hand, if $O_x$ is the local operator in~\eqref{eq:localop}, then 
$\Var[\partial_\nu C^L]=\frac{1}{8n^2}$.

These results suggest that DQNNs with simple shallow local perceptrons and global cost functions are untrainable when randomly initialized. On the other hand, they also indicate that barren plateaus for DQNNs might be avoided by employing: (1) shallow (local) perceptrons, and (2) local cost functions.

Let us now consider the shallow local perceptron of Fig.~\ref{fig:f3}(b), where each unitary $W$ forms a local $2$-design on two qubits. For this  architecture the ensuing DQNN can be {\it exactly} mapped into a layered hardware-efficient ansatz as in~\cite{cerezo2020cost}, where two layers of the DQNN correspond to a single layer of the hardware-efficient ansatz~\cite{kandala2017hardware}. Note that this mapping is not general, but rather valid for the specific architecture in Fig.~\ref{fig:f3}(b). As shown in Ref.~\cite{cerezo2020cost}, when employing a global cost function, with $O_x$ given by~\eqref{eq:globalop}, one finds that if the number is layers is $\OC(\poly(\log(n)))$, then the DQNN cost function exhibits barren plateaus as
\begin{equation}
    \Var[\partial_\nu C^G]\leq \widehat{f}(n)
    \,, \,\,\, \text{with}\quad \widehat{f}(n)\in \OC\left((\sqrt{3}/4)^n\right)\,.
\end{equation}
On the other hand, for a local cost function with $O_x$ given by~\eqref{eq:localop}, if the number of layers is in $\OC(\log(n))$, then there is no barren plateau~\cite{cerezo2020cost} as
\begin{equation}\label{eq:lowerb}
     \widehat{g}(n)\leq \Var[\partial_\nu C^L]
    \,, \,\,\,\text{with}\quad \widehat{g}(n)\in \Omega\left(1/\poly(n)\right)\,.
\end{equation}
Here we remark that \eqref{eq:lowerb} was obtained following the  same assumptions as those used in Corollary~2 of~\cite{cerezo2020cost}.
Note that obtaining a lower bound for the variance implies that the DQNN trainability is guaranteed.

\bigskip

{\it Parameter matrix multiplication.---}While in random parametrized quantum circuits one optimizes and trains a single gate angle at a time, other optimization approaches can also be considered. In what follows we analyze the trainability for a method introduced in Ref.~\cite{beer2020training} where at each time-step all perceptrons are simultaneously optimized.

In this training approach, which we call parameter matrix multiplication, the perceptrons are not explicitly decomposed into quantum circuits, but rather are treated as unitary matrices. The perceptrons $V_j^l(0)$ are randomly initialized at time-step zero, and at each step $s$ they are updated via 
\begin{equation}\label{eq:update}
    V_j^l(s+\varepsilon)=e^{i\varepsilon H_j^l(s)}V_j^l(s)\,.
\end{equation}
The matrices $H_j^l$ are such that $\Tr[(H_j^l)^2]\leq 2^{n+1}$ and are parametrized as 
$H_j^l(s)=\sum_{\vec{u}\vec{v}}h^l_{j,\vec{u},\vec{v}} X^{\vec{u}}Z^{\vec{v}}$, with $X^{\vec{u}}Z^{\vec{v}}=X^{u_1}_1Z^{v_1}_1\otimes X^{u_2}_{2}Z^{v_2}_{2}\ldots$, and where $X_j$ and $Z_j$ are Pauli operators on qubit $j$. The matrices $K_j^l(s)$ are called {\it parameter matrices}, and at each time-step the coefficients $h^l_{j,\vec{u},\vec{v}}$ need to be optimized. As shown in the Supplemental Material, if at least one perceptron $V_j^l(0)$ is sufficiently random so that it forms a global unitary $1$-design, then we find $\langle \partial C/\partial s \rangle\equiv\langle \partial_s C \rangle=0$.

As proved in the Supplemental Material, the following theorem holds. 
\begin{theorem}\label{theo1}
Consider a DQNN with deep global perceptrons, which are updated via the parameter matrix multiplication of~\eqref{eq:update}. Suppose that for all $j,l$, the $V_j^l(0)$ perceptrons form independent $2$-designs over $n+1$ qubits. Then the variance of the partial derivative of the cost function with respect to the time-step parameter $s$ is upper bounded as
\begin{equation}
\Var[\partial_s C]\leq f(n), \quad \text{with} \quad f(n)\in\OC\left(1/2^{n}\right)\,,
\end{equation}
when $O_x$ is the global operator of~\eqref{eq:globalop}, or the local operator in~\eqref{eq:localop}. 
\end{theorem}

Although the updating method in \eqref{eq:update} simultaneously updates all perceptrons at each time-step, Theorem  \ref{theo1} implies that barren plateaus also arise when using the parameter matrix multiplication method.

We note that our proof techniques invoke the pure state properties of input and output states. Since the output state of a randomly initialized DQNN will be close to a maximally mixed across any bipartite cut \cite{marrero2020entanglement}, we speculate that our results can be extended to expectation values of the arbitrary Hamiltonian. We leave this question for future work.


\bigskip

{\it Conclusions.---}In this Letter, we analyzed the trainability of a special class of QNNs called DQNNs. We first proved that the trainability  of DQNNs is not always guaranteed as they can exhibit  barren plateaus in their cost function landscape. The existence of such barren plateaus was linked to the localities (i.e., the number of qubits they act non-trivially on) of the perceptrons and of the cost function. Specifically, we showed that: (i) DQNNs with deep global perceptrons are untrainable despite the dissipative nature of the architecture, and (ii) for shallow and local perceptrons, employing global cost functions leads to barren plateaus, while using local costs avoids them. We note that our results are completely general for DQNN architectures, e.g., covering arbitrary numbers of hidden layers and general perceptrons acting on any number of qubits.

In addition, we provided a specific architecture for DQNNs with local shallow perceptrons that can be exactly mapped to a layered hardware-efficient ansatz. This result not only indicates that hardware-efficient QNNs can be represented as DQNNs, but it also allows us to derive trainability guarantees for these DQNNs. In this case, since the  perceptrons are local, each neuron only receives information from a small number of qubits in the previous layer. Such architecture is reminiscent of classical convolutional neural networks, which are known to avoid some of the trainability problems of fully connected networks~\cite{aloysius2017review}. 

These results show  that much work needs to be done to understand the trainability of QNNs and guarantee that they can provide a quantum speedup over classical neural networks. For instance,  interesting future research directions are QNN-specific optimizers~\cite{stokes2019quantum,kubler2019adaptive,koczor2019quantum,arrasmith2020operator},  analyzing the resilience of QNNs to noise~\cite{mcclean2016theory,sharma2019noise}, and strategies to prevent barren plateaus~\cite{larose2019variational,grant2019initialization,verdon2019learning,tjvpat20}. Another interesting direction is to extend our results to the case when the input and output states are mixed states, particularly when the goal is to match marginals of the target output state and the output of a DQNN \cite{bolens2021reinforcement}.  Furthermore, exploring architectures beyond DQNNs and hardware-efficient QNNs would be of interest, particularly if such architectures have large-scale trainability.

\bigskip

\begin{acknowledgments}
We thank Jarrod McClean, Tobias Osborne, and Andrew Sornborger for helpful conversations. All authors acknowledge support from LANL's Laboratory Directed Research and Development (LDRD) program.  MC was also supported by the Center for Nonlinear Studies at LANL. PJC also acknowledges support from the LANL ASC Beyond Moore's Law project. This work was also supported by the U.S. Department of Energy (DOE), Office of Science, Office of Advanced Scientific Computing Research, under the Accelerated Research in Quantum Computing (ARQC) program.
\end{acknowledgments}

\textit{Supplemental Material.---}The Supplemental Material contains details of our proofs and References~\cite{puchala2017symbolic,Fukuda_2019}.

\textit{Note Added.---}Our work is the first to analyze barren plateaus in the context of data science applications, and also the first to consider perceptron-based quantum neural networks (QNNs). Our work has inspired more recent studies of trainability for other QNN architectures, such as quantum convolutional neural networks~\cite{pesah2020absence}, tree-based architectures~\cite{zhang2020toward}, and others~\cite{zhao2021analyzing,wang2020noise,abbas2020power}. We also note that our results can also be interpreted as a type of entanglement-induced barren plateau. Here, a large amount of entanglement in a parameterized quantum circuit can lead to trainability issues when qubits are discarded, and the output qubits are concentrated around the maximally mixed state. This phenomenon was further studied in~\cite{marrero2020entanglement,patti2021entanglement}.


\bibliography{ref.bib}

\onecolumngrid

\pagebreak

\appendix

\vspace{0.5in}

\setcounter{theorem}{0}

\begin{center}
	{\Large \bf Supplemental Material for {\it Trainability of Dissipative Perceptron-Based Quantum Neural Networks}}
\end{center}

Here we provide proofs for the main results and theorems of the manuscript {\it Trainability of Dissipative Perceptron-Based Quantum Neural Networks}.   In Section~\ref{sec:preliminaries} we first present useful definitions and lemmas that will be employed to derive the main results. Then in Section~\ref{sec:average} we show that $\langle\partial C\rangle =0$ for the Dissipative Quantum Neural Networks (DQNN) considered in the main text. Finally, in  Sections~\ref{sec:prooftheo2} and~\ref{sec:proof-theo1} we provide proofs for Theorem~\ref{theo2SM} and Theorem~\ref{theo1SM}, respectively. We note that we first provide a proof for Theorem~\ref{theo1SM}, since the proof of Theorem~\ref{theo2SM} can be built from the latter.

We also note that for our proofs we will assume a DQNN such that each perceptron only acts on one output qubit, and where there are no hidden layers. We generalize to DQNNs acting on $m$ output qubits and with $L$  hidden layers in Sections~\ref{sec:dqnnnm} and~\ref{sec:dqnn-hidden}, respectively.

\section{Preliminaries}\label{sec:preliminaries}

\medskip

\noindent \textbf{Properties of the Haar measure.} Let $d\mu_H(V)\equiv d\mu(V)$ be the volume element of the Haar measure, with $V\in U(d)$, and where $U(d)$ denotes the unitary group of degree $d$. Then the following properties  hold:
\begin{itemize}
    \item The volume of the Haar measure is finite:
    \begin{equation}
        \int_{U(d)}d\mu(V)<\infty\,.
    \end{equation}
    \item The Haar measure is left- and right-invariant under the action of the unitary group of degree $d$. That is, for any integrable function g(V) and for any $W\in U(d)$ we have 
    \begin{equation}
        \int_{U(d)}d\mu(V)g(WV)=\int_{U(d)}d\mu(V)g(VW)=\int_{U(d)}d\mu(V)g(V)\,.
    \end{equation}
        \item The Haar measure is uniquely defined up to a multiplicative constant factor. Let $d\omega(V)$ be an invariant measure, then, there existe a constant $c$ such that 
    \begin{equation}
        d\omega(V)= c\cdot d\mu(V)\,.
    \end{equation}        
     
\end{itemize}

\medskip

\noindent \textbf{Definition: $t$-design.} Let be $\{V_y\}_{y\in Y}$,  of size $|Y|$, be a set of unitaries $V_y$ acting on a $d$-dimensional Hilbert space. In addition, let $P_{t}(W)$ be a polinomial of degree at most $t$ in the matrix elements of $V$, and at most $t$ in those of  $V\ad$. Then $\{V_y\in U(d)\}_{y\in Y}$ is a unitary $t$-design if for every $P_{t}(W)$, the following holds~\cite{dankert2009exact}:
\begin{equation}
    \frac{1}{|Y|}\cdot \sum_{y\in Y}P_{t}(V_y)=\int d\mu(V) P_{t}(V)\,, \label{eq:SMt-design}
\end{equation}
where it is implicit that the integral is over $U(d)$.

\medskip

\noindent \textbf{Symbolic integration.} Here we recall formulas which allow for the symbolical integration with respect to the Haar measure on a unitary group~\cite{puchala2017symbolic}. For any $V\in U(d)$ the following expressions are valid for the first two moments:  
\small
\begin{equation}\label{eq:delta}
\begin{aligned}
    \int d\mu(V)v_{\vec{i}\vec{j}}v_{\vec{p}\vec{k}}^*&=\frac{\delta_{\vec{i}\vec{p}}\delta_{\vec{j}\vec{k}}}{d}\,,   \\
\int d\mu(V)v_{\vec{i}_1\vec{j}_1}v_{\vec{i}_2\vec{j}_2}v_{\vec{i}_1'\vec{j}_1'}^{*}v_{\vec{i}_2'\vec{j}_2'}^{*}&=\frac{\delta_{\vec{i}_1\vec{i}_1'}\delta_{\vec{i}_2\vec{i}_2'}\delta_{\vec{j}_1\vec{j}_1'}\delta_{\vec{j}_2\vec{j}_2'}+\delta_{\vec{i}_1\vec{i}_2'}\delta_{\vec{i}_2\vec{i}_1'}\delta_{\vec{j}_1\vec{j}_2'}\delta_{\vec{j}_2\vec{j}_1'}}{d^2-1}
-\frac{\delta_{\vec{i}_1\vec{i}_1'}\delta_{\vec{i}_2\vec{i}_2'}\delta_{\vec{j}_1\vec{j}_2'}\delta_{\vec{j}_2\vec{j}_1'}+\delta_{\vec{i}_1\vec{i}_2'}\delta_{\vec{i}_2\vec{i}_1'}\delta_{\vec{j}_1\vec{j}_1'}\delta_{\vec{j}_2\vec{j}_2'}}{d(d^2-1)}\,,
\end{aligned}
\end{equation}
\normalsize
where $v_{\vec{i}\vec{j}}$ are the matrix elements of $V$. Assuming $d=2^n$, we use the notation $\vec{i} = (i_1, \dots i_n)$ to denote a bitstring of length $n$ such that $i_1,i_2,\dotsc,i_{n}\in\{0,1\}$. 

\medskip

\noindent \textbf{Useful identities.} From Eqs.~\eqref{eq:delta}, the following identities can be readily derived~\cite{cerezo2020cost}
\small
\begin{align}
    \int d\mu(V)\Tr\left[VAV\ad B\right]&=\frac{\Tr\left[A\right]\Tr\left[B\right]}{d}
    \label{eq:lemma1}\\
    \int d\mu(V)\Tr[VAV\ad BVCV\ad D]&=\frac{\Tr[A]\Tr[C]\Tr[BD]+\Tr[AC]\Tr[B]\Tr[D]}{d^2-1}-\frac{\Tr[AC]\Tr[BD]+\Tr[A]\Tr[B]\Tr[C]\Tr[D]}{d(d^2-1)}\label{eq:lemma2}\\
    \int d\mu(V)\Tr[VAV\ad B]\Tr[VCV\ad D]&=\frac{\Tr[A]\Tr[B]\Tr[C]\Tr[D]+\Tr[AC]\Tr[BD]}{d^2-1}-\frac{\Tr[AC]\Tr[B]\Tr[D]+\Tr[A]\Tr[C]\Tr[BD]}{d(d^2-1)}\label{eq:lemma3}
\end{align}
\normalsize
where $A,B,C,D$ are linear operators on a $d$-dimentional Hilbert space. We point readers to \cite{cerezo2020cost} for detailed proofs of the aforementioned identities.

In addition, let us consider a  bipartite Hilbert space $\mathcal{H} \equiv \mathcal{H}_{1}\otimes \mathcal{H}_{2}$ of dimension $\dim(\mathcal{H})=d_1 d_2$. Let $A, B: \HC \to \HC$ be linear operators. 
Then the following equality holds:
\begin{equation}
\int d\mu(V)(\id_{1}\otimes V)A( \id_{1}\otimes V\ad)B= \frac{\Tr_{2}\left[A\right]\otimes\id_2}{d_2}B\,,
    \label{eq:lemma4}
\end{equation}
where $\Tr_2$ indicates the partial trace over $\mathcal{H}_2.$ We note that Eq.~\eqref{eq:lemma4} follows by employing \eqref{eq:delta}.

\medskip

\begin{lemma}\label{lemma1}
Let $\mathcal{H}\equiv \mathcal{H}_{1}\otimes \mathcal{H}_{2}$ denote a bipartite Hilbert space of dimension $\dim(\mathcal{H})=d_1 d_2$, such that $d_1=2^n$ and $d_2=2^{n'}$. Let $A,B:\mathcal{H}\rightarrow \mathcal{H}$ be linear operators, and let $V\in U(d)$. Then we have that 
\begin{align}
\Tr\left[(\id_{1}\otimes V)A(\id_{1}\otimes V\ad) B\right]=\sum_{\vec{p},\vec{q}}\Tr\left[V A_{\vec{q}\vec{p}}V\ad B_{\vec{p}\vec{q}}\right]\,.
    \label{eq:lemma5}   
\end{align}
 Where the summation runs over all bitstrings of length $n'$, and where we define
\begin{equation}
    A_{\vec{q}\vec{p}}=\Tr_{1}\left[(\ketbra{\vec{p}}{\vec{q}}\otimes\id_2)A\right]\,, \quad B_{\vec{p}\vec{q}}=\Tr_{1}\left[(\ketbra{\vec{q}}{\vec{p}}\otimes\id_2)B\right]\, .
\end{equation}
Here $\id_{1}$ ($\id_{2}$) is the identity operator over $\mathcal{H}_{1}$ ($\mathcal{H}_{2}$). 
\end{lemma}
Equation~\eqref{eq:lemma5}  can be derived by expanding the operators in the computational basis as derived in~\cite{cerezo2020cost}.

\begin{lemma}\label{lemma2}
Let $\mathcal{H}\equiv \mathcal{H}_{1}\otimes \mathcal{H}_{2}\otimes \mathcal{H}_{3}\otimes \mathcal{H}_{4}$ denote a  Hilbert space of dimension $\dim(\mathcal{H})=d_1 d_2 d_3d_4$. Consider the following linear operators $H:\mathcal{H}_1\otimes\mathcal{H}_2\rightarrow\mathcal{H}_1\otimes\mathcal{H}_2$, $K:\mathcal{H}_1\otimes\mathcal{H}_4\rightarrow\mathcal{H}_1\otimes\mathcal{H}_4$,  $S,S':\mathcal{H}\rightarrow\mathcal{H}$, $P,P':\mathcal{H}_3\otimes\mathcal{H}_4\rightarrow\mathcal{H}_3\otimes\mathcal{H}_4$, $V:\mathcal{H}_1\otimes\mathcal{H}_3\rightarrow\mathcal{H}_1\otimes\mathcal{H}_3$, and $U:\mathcal{H}_1\otimes\mathcal{H}_4\rightarrow\mathcal{H}_1\otimes\mathcal{H}_4$. Then, let us define the following operators acting on  $\mathcal{H}_{1}\otimes \mathcal{H}_{2}$:
\begin{equation}
    M=\Tr_{34}[P[VUSU\ad V\ad,H]]\,,\quad \quad B=\Tr_{34}[P'VU[S',K]U\ad V\ad]\,,
\end{equation}
where $\Tr_{34}$ indicates the trace over subsystems $\mathcal{H}_{3}\otimes \mathcal{H}_{4}$. The following equality holds:
\begin{equation}\label{eq:proof}
\int d\mu(V)\Tr_{12}\left[MB\right]=0\,.
\end{equation}
\end{lemma}

\begin{lemma}\label{lemma3}
Let $\mathcal{H}\equiv \mathcal{H}_{1}\otimes \mathcal{H}_{2}\otimes \mathcal{H}_{3}\otimes \mathcal{H}_{4}$ denote a  Hilbert space of dimension $\dim(\mathcal{H})=d_1 d_2 d_3d_4$. Consider the following linear operators
$V:\mathcal{H}_1\otimes \mathcal{H}_2\rightarrow \mathcal{H}_1\otimes \mathcal{H}_2$,   
$U:\mathcal{H}_1\otimes \mathcal{H}_3\otimes \mathcal{H}_4\rightarrow \mathcal{H}_1\otimes \mathcal{H}_3\otimes \mathcal{H}_4$,  
$P,P':\mathcal{H}_2\otimes \mathcal{H}_3\otimes \mathcal{H}_4\rightarrow \mathcal{H}_2\otimes \mathcal{H}_3\otimes \mathcal{H}_4$,  and $Z,Z':\mathcal{H}_4\rightarrow \mathcal{H}_4$. 
Let $P$  and $P'$ be tensor product operators of the form $P=\Pi \otimes \widetilde{P}$, and $P'=\Pi' \otimes \widetilde{P}'$,  where $\widetilde{P},\widetilde{P}': \mathcal{H}_3\otimes \mathcal{H}_4\rightarrow \mathcal{H}_3\otimes \mathcal{H}_4$, and where $\Pi, \Pi':\mathcal{H}_2\rightarrow \mathcal{H}_2 $ are rank one projector such that  $\Pi^2=\Pi$, $\Tr[\Pi]=1$, $(\Pi')^2=\Pi'$, and  $\Tr[\Pi']=1$. Then, consider the following operators acting on subsystem $\mathcal{H}_1$:
\begin{align}
\Omega &= \Tr_{2,3,4}[ P V U Z U\ad V\ad]\,, \quad \widetilde{\Omega} = \Tr_{3,4}[ \widetilde{P} U Z U\ad ]\,,\\
\Omega' &= \Tr_{2,3,4}[ P' V U Z' U\ad V\ad]\,, \quad \widetilde{\Omega}' = \Tr_{3,4}[ \widetilde{P}' U Z' U\ad ]\,,
\end{align}
where $\Tr_{2,3,4}$ $(\Tr_{3,4})$ indicates the trace over  $\mathcal{H}_2\otimes\mathcal{H}_3\otimes\mathcal{H}_4$ $(\mathcal{H}_3\otimes\mathcal{H}_4)$. 
Then, the following equalities hold 
\begin{align}\label{eq:proof2}
\int d\mu(V)\Tr_{1}\left[\Omega \Omega' \right]&=\frac{d_1^2 d_2}{d_1^2 d_2^2 -1}\left(\Tr[\Pi \Pi']-\frac{1}{d_1^2d_2}\right)\Tr_1[\widetilde{\Omega}\widetilde{\Omega}']+\frac{d_1 d_2^2}{d_1^2 d_2^2 -1}\left(1-\frac{\Tr[\Pi \Pi']}{d_2}\right)
\Tr_1[\widetilde{\Omega}]\Tr[\widetilde{\Omega}']\,.\\
\int d\mu(V)\Tr_{1}\left[\Omega]\Tr[ \Omega' \right]&=\frac{d_1 d_2}{d_1^2 d_2^2 -1}\left(\Tr[\Pi \Pi']-\frac{1}{d_2}\right)\Tr_1[\widetilde{\Omega}\widetilde{\Omega}']+\frac{d_1^2 d_2^2}{d_1^2 d_2^2 -1}\left(1-\frac{\Tr[\Pi \Pi']}{d_1^2 d_2}\right)
\Tr_1[\widetilde{\Omega}]\Tr[\widetilde{\Omega}']\,.
\end{align}
\end{lemma}

Lemmas~\ref{lemma2}, and ~\ref{lemma3} can be derived by explicitly integrating over $V$ using~\eqref{eq:delta}. In addition, one can also use the RTNI package for symbolic integrator over Haar-random tensor networks of Ref.~\cite{Fukuda_2019}.

\bigskip

{\bf Remark--} 
We note that for simplicity we derive the proofs of our theorems for the case when the output states $\ket{\phi^\tout_x}$ in the training set are tensor product of computational basis states over $n$ qubits. The proofs can then be trivially generalized for arbitrary tensor-product states over $n$ qubits. Again, for simplicity, we provide rigorous proofs for the case when there are no hidden layers and later argue in Section~\ref{sec:dqnn-hidden} that our results hold for DQNNs with hidden layers. We also note that DQNNs considered in our work have a simple structure where unitaries act on $n+1$ qubits. In Section~\ref{sec:dqnnnm} we generalize our results to the case when unitaries act on $n+m$ qubits.

\section{Generalization of our results to quantum machine learning task}
In the main text we have stated our main results in terms of state preparation, where the training set is of the form $\{\ket{\phi_x^{\tin}},\ket{\phi_x^{\tout}}\}$. As we show here, this result can be used to generalize our result for other supervised learning quantum machine learning tasks.

Consider now the case when the training set is of the form $\{\ket{\psi_x^{\text{in}}},y_x\}$. Here $y_x$ are labels associated with each input quantum state $\ket{\psi_x^{\text{in}}}$, and the goal of the DQNN is to predict a label $\widetilde{y}_x$ that matches the true label. For simplicity we here consider the case of binary classification $y_x\in\{-1,1\}$.

Following the scheme introduced in~\cite{havlivcek2019supervised}, the predicted labels $\widetilde{y}_x$ can be obtained by performing a binary measurement $M_y$ on the states $\rho_x^{\text{out}}$. Without loss of generality, this measurement is taken to be on the $z$-basis, so that the measurement outcomes are bitstrings $\vec{z}$ of length $k$, where $k$ is the number of measured qubits. The measurement operator is then given by 
\begin{equation}
    M_y=\frac{\id+y\vec{h}}{2}\,, \quad \text{with} \quad \vec{h}=\sum_{\vec{z}}h(\vec{z})\dya{\vec{z}}\,.
\end{equation}
Thus, the probability of obtaining label $y$ from input state $\ket{\psi_x^{\text{in}}}$ is
\begin{equation}
    p_y(\ket{\psi_x^{\text{in}}})=\sum_{\vec{z}}h(\vec{z})\Tr[O_{\vec{z}} \rho^{\text{out}}_x]\,,
\end{equation}
where
\begin{equation}\label{eq:cost-QML}
    O_{\vec{z}}=\id\otimes\dya{\vec{z}}\,.
\end{equation}
From these probabilities, the assigned labels are $\widetilde{y}_x=p_y(\ket{\psi_x^{\text{in}}})\geq p_{-y}(\ket{\psi_x^{\text{in}}})$.

Note that Eq.~\eqref{eq:cost-QML} is precisely of the form considered in the main text. Thus, the formalism in the main text can be directly extended for other supervised learning tasks.

\section{Proof of $\langle \partial C\rangle =0$}
\label{sec:average}

In this section, we prove that the average value of the partial derivative of the cost function $\langle \partial C \rangle$ is not biased towards any particular value. In particular, we show that $\langle \partial C \rangle = 0$. We prove this result for both cases: 1) random parameterized quantum circuits and 2) parameter matrix multiplication method.

\bigskip

\noindent \textbf{Random parametrized quantum circuits.} Let us  first consider the case when the perceptrons are random parametrized quantum circuits of the form
\begin{equation}\label{eq:parametrizationSM}
V_j^l(\thv_j^l)=\prod_{k=1}^{\eta_j^l}R_k(\theta^k)W_k\,, 
\end{equation}
with  $W_k$ an unparametrized unitary, $R_k(\theta^k)=e^{-(i/2) \theta^k \Gamma_k}$, and where $\Gamma_k$ is a Hermitian operator with $\Tr[ \Gamma_k^2]\leq 2^{n+1}$.

We recall from the main text that the training set consisting of input and output pure quantum states: $\{\ket{\phi_x^{\tin}},\ket{\phi_x^{\tout}}\}_{x=1}^N$. As mentioned in the Remark in the previous section, we now assume that  $\ket{\phi_x^{\tout}}$ are computational basis states, i.e.,  $\ket{\phi_x^{\tout}}\equiv \ket{\vec{z}^x}=\ket{z^x_1z^x_2\ldots z^x_n}$. The DQNN cost function is defined as
\begin{equation}\label{eq:costxSM}
    C=\frac{1}{N}\sum_{x=1}^N C_x\,, \quad \text{with} \quad C_x =\Tr[O_x \rho^\tout_x]\,,
\end{equation}
where $O_x$ is given by the global observable
\begin{equation}\label{eq:globalopSM}
    O_x^G\equiv\id - \ketbra{\phi_x^{\tout}}{\phi_x^{\tout}}\,
\end{equation}
or by the local operator
\begin{equation}\label{eq:localopSM}
    O_x^L=\id - \frac{1}{n}\sum_{i=1}^n\dya{z^x_i}\otimes \id_{\overline{i}}\,,
\end{equation}
where $\id_{\overline{i}}$ indicates identity on all qubits in the output layer except for qubit $i$.
Then the partial derivative of $C$ with respect to parameter $\theta^\nu$ in a perceptron $V_j^l$ is given by
\begin{align}
\partial_\nu C&= \frac{i}{2N}\sum_{x=1}^{N}\Tr\Big[ A_j^l\tilde{\rho}^\tin_x (A_j^l)\ad[\id_{\overline{j}}^{\overline{l}}\otimes\Gamma_k,(B_j^l)\ad\tilde{O}_x B_j^l]\Big]\,,\\
&= -\frac{i}{2N}\sum_{x=1}^{N}\Tr\Big[ (B_j^l)\ad \tilde{O}_x B_j^l[\id_{\overline{j}}^{\overline{l}}\otimes\Gamma_k,A_j^l\tilde{\rho}^\tin_x (A_j^l)\ad]\Big]\,,\label{eq:costparderA}
\end{align}
where the trace is taken all qubits in the QNN, and where $\id_{\overline{j}}^{\overline{l}}$ indicates the identity on all qubits on which $V_j^l$ does not act on. Moreover, we have defined 
\begin{align}
A_j^l&=\id_{\overline{j}}^{\overline{l}}\otimes\prod_{k=\nu}^{\eta_j^l}R_k(\theta^k)W_k,\label{eq:AandB1}\\
    B_j^l&=\id_{\overline{j}}^{\overline{l}}\otimes\prod_{k=1}^{\nu-1}R_k(\theta^k)W_k\,,\label{eq:AandB2}\\
\tilde{\rho}^\tin_x&=V_{j-1}^l\ldots V_1^1 (\rho^\tin_x\otimes\dya{\vec{0}}_{\thid,\tout}) (V_1^1)\ad\ldots (V_{j-1}^l)\ad\,,\\
\tilde{O}_x&=(V_{j+1}^l)\ad\ldots (V_{n_\tout}^\tout)\ad (\id_{\tin,\thid}\otimes O_x) V_{n_\tout}^\tout\ldots V_{j+1}^l\,.
\end{align}

Let us first consider the case when $B_j^l$ is a $1$-design. Assuming that all the perceptrons are independent, and that $A_j^l$ and $B_j^l$ are also independent, we can  express
\begin{equation}\label{eq:averages}
    \langle\cdots\rangle=\langle\cdots\rangle_{V_1^1,\ldots,A_j^l,B_j^l,\ldots,V_{n_\tout}^\tout}=\langle\langle\cdots\rangle_{B_j^l}\rangle_{V_1^1,\ldots,A_j^l,\ldots,V_{n_\tout}^\tout}\,.
\end{equation}
Hence, we can first compute the average of $\partial_\nu C$ over $B_j^l$ as
\begin{align}
   \langle \partial_\nu C\rangle_{B_j^l}=& \frac{i}{2N}\sum_{x=1}^{N}\Tr\Big[ A_j^l\tilde{\rho}^\tin_x (A_j^l)\ad[\id_{\overline{j}}^{\overline{l}}\otimes\Gamma_k,\int d\mu(B_j^l)(B_j^l)\ad\tilde{O}_x B_j^l]\Big]\label{eq:SM_gradWAzero0}\\
    =&\frac{i}{2N}\sum_{x=1}^{N}   \Tr\left[A_j^l\tilde{\rho}^\tin_x (A_j^l)\ad[\id_{\overline{j}}^{\overline{l}}\otimes\Gamma_k,\frac{1}{2^{m+1}}\Tr_{jl}[\tilde{O}_x]\otimes\id_j^l\right]\Big]\nonumber\\
    =&0\,.
    \label{eq:SM_gradWAzero}
\end{align}
Here, $\Tr_{jl}$ indicates the trace over the $m+1$ qubits on which $B_j^l$ acts, and  $\id_j^l$ ($\id_{\overline{j}}^{\overline{l}}$) is the identity operator  over the qubits on which $B_j^l$ acts (does not act). 
For the second equality we used \eqref{eq:lemma4}, and  in the third equality we used the fact that a commutator inside the trace is always zero. With a similar argument it is straightforward to show from~\eqref{eq:costparderA} that    $\langle \partial_\nu C\rangle_{A_j^l}=0$ if $A_j^l$ is a $1$-design. Finally, from Eq.~\eqref{eq:averages} we know that   $\langle \partial_\nu C\rangle_{A_j^l}=0$ ( $\langle \partial_\nu C\rangle_{B_j^l}=0$) leads to  $\langle \partial_\nu C\rangle=0$.

\medskip

\noindent \textbf{Parameter matrix multiplication.} In this case the perceptrons are updated with the parameter multiplication matrix method of  Ref.~\cite{beer2020training}. We recall that  the perceptrons $V_j^l(0)$ are randomely initialized, and then at each step $s$ they are updated via 
\begin{equation}\label{eq:updateSM}
    V_j^l(s+\varepsilon)=e^{i\varepsilon H_j^l(s)}V_j^l(s)\,,
\end{equation}
where $H^l_j(s)$ is a Hermitian operator such that $\Tr[(H_j^l)^2]\leq 2^{n+1}$.

As shown in~\cite{beer2020training}, the derivative of $C(s)$ with respect to $s$ is given by,
\begin{align}
    \partial_s C(s)&= \lim_{\varepsilon \rightarrow 0}\frac{C(s+\varepsilon)-C(s)}{\varepsilon}\nonumber\\ &=\frac{i}{N}\sum_{x=1}^N\left[\sum_{l=1}^{\tout}\sum_{j=1}^{n_l} \Tr\Big[\id_{\overline{j}}^{\overline{l}}\otimes H^l_{j}(s)[V_j^l(s) \tilde{\rho}^\tin_x (s) (V_j^l(s))\ad,\tilde{O}_x(s)]\Big]\right]\label{eq:partialsC}\\
    &=\frac{i}{N}\sum_{x=1}^N\left[\sum_{l=1}^{\tout}\sum_{j=1}^{n_l} \Tr\Big[\tilde{O}_x(s)[ \id_{\overline{j}}^{\overline{l}}\otimes H^l_{j}(s),V_j^l(s) \tilde{\rho}^\tin_x(s)  (V_j^l(s))\ad]\Big]\right]\,,
\end{align}
where now
\begin{align}
\tilde{\rho}^\tin_x(s)&=V_{j-1}^l(s)\ldots V_1^1(s) (\rho^\tin_x\otimes\dya{\vec{0}}_{\thid,\tout}) (V_1^1(s))\ad\ldots (V_{j-1}^l(s))\ad\,,\label{eq:rhotildin}\\
    \tilde{O}_x(s)&=(V_{j+1}^l(s))\ad\ldots (V_{n_\tout}^\tout(s))\ad (\id_{\tin,\thid}\otimes O_x) V_{n_\tout}^\tout(s)\ldots V_{j+1}^l(s)\,.\label{eq:rhotildout}
\end{align}
Let us now consider time-step $s=0$. If we assume that all the perceptrons are independently initialized we have 
\begin{equation}\label{eq:averagesmult}
    \langle\cdots\rangle=\langle\cdots\rangle_{V_1^1(0),\ldots,V_{n_\tout}^\tout(0)}=\langle\langle\cdots\rangle_{V_j^l(0)}\rangle_{V_1^1(0),\ldots,V_{j-1}^l(0),V_{j+1}^l(0),\ldots,V_{n_\tout}^\tout(0)}\,.
\end{equation}
Therefore, if $V_j^l(0)$ is a $1$-design, from \eqref{eq:SM_gradWAzero0}--\eqref{eq:SM_gradWAzero} it follows that $\langle \partial_s C \rangle_{V_j^l(0)}=\langle \partial_s C \rangle=0$.

\setcounter{theorem}{1}

\section{Proof of Theorem \ref{theo1SM}}\label{sec:proof-theo1}
In this section, we provide a proof of Theorem \ref{theo1SM}. In what follows we will assume a DQNN with no hidden layers, and such that each perceptron only acts on one output qubit. We generalize to DQNNs acting on $m$ output qubits and with $L$  hidden layers in sections ~\ref{sec:dqnnnm} and~\ref{sec:dqnn-hidden}, respectively.  
 
 We first recall Theorem \ref{theo1SM} for convenience: 
\begin{theorem}\label{theo1SM}
Consider a DQNN with deep global perceptrons, which are updated via the parameter matrix multiplication of~\eqref{eq:updateSM}, and such that  $V_j^l(0)$ form independent $2$-designs over $n+1$ qubits. Then the variance of the partial derivative of the cost function with respect to the time-step paramater $s$ is upper bounded as
\begin{equation}\label{eq:theo1eq}
\Var[\partial_s C]\leq f(n), \quad \text{with} \quad f(n)\in\OC\left(1/2^{n}\right)\,,
\end{equation}
when $O_x$ is the global operator as in \eqref{eq:globalopSM} or the local operator as  in~\eqref{eq:localopSM}. 
\end{theorem}
\begin{proof}

We divide our proof in several subsections consisting of different cases. We also analyze the global and local cost functions separately. For simplicity we consider a DQNN with no hidden layers, and where both input and output layers consist of $n$ qubits. 
Moreover, we  denote the randomly initialized perceptrons at time step $s=0$ as $V_j$, where we also henceforth drop the superscript index that denotes the layer.  

Let
\begin{align}
    \sigma_x^{\tin} &= \ketbra{\phi_x^{\tin}}{\phi_x^{\tin}}\otimes \ketbra{\vec{0}}{\vec{0}}_{\text{out}}\,,\quad \quad 
    \sigma_x^{\tout}  =  \id_{\tin}\otimes O_x\,, \label{eq:sigmas}\\
    G_j^x &= [V_j \cdots V_1 \sigma_x^{\tin} V_1\ad \cdots V_j\ad, V_{j+1}\ad \cdots V_n\ad \sigma_x^{\tout}V_n \cdots V_{j+1}]. 
\end{align}
Then \eqref{eq:partialsC} can be rewritten as 
\begin{align}\label{eq:partialCsum}
    \partial_s C = \frac{i}{N}\sum_{x=1}^N\left[\sum_{j=1}^n \Tr[G^x_j H_j]\right].
\end{align}
Note that from the cyclicity of the trace, each term $\Tr[G^x_j H_j]$ can also be expressed as
\begin{align}
    \Tr[G^x_j H_j] = \Tr\left[V_j \cdots V_1 \sigma_x^{\tin}V_1\ad \cdots V_j\ad \left[V_{j+1}\ad\cdots V_n\ad \sigma_x^{\tout}V_n \cdots V_{j+1}, H_j\right]\right]. \label{eq:gxj-hxj-gen-form}
\end{align}

The proof of Theorem~\ref{theo1SM} is constructed as follows. We first note that $\langle \partial C \rangle =0$ and thus, $\text{Var}[\partial_s C]$ only depends on the second moment of partial derivatives. Moreover, $\langle (\partial_s C)^2\rangle$ depends on terms of the following form $\Tr[G^x_j H_j]\Tr[G^{x^{\prime}}_{j^{\prime}} H_{j^{\prime}}]$. We first consider a single term in the summation over $j$ and $x$ in~\eqref{eq:partialCsum}, and show that $\langle (\Tr[G^x_j H_j])^2 \rangle \leq f(n)$ with $f(n)$ as in~\eqref{eq:theo1eq}.  Then we prove that the cross terms in $x$ and $j$ also satisfy $\langle \Tr[G^x_j H_j] \Tr[G^{x'}_{j'} H_{j'}] \rangle \leq f(n)$.

\subsection{Global Cost}\label{sec:gc-thm1}

\subsubsection{Fixed $j$ and fixed $x$}\label{sec:fixed-jx-thm1-gc}

Let us first consider the case when the cost function is defined in terms of the global operator in~\eqref{eq:globalopSM}. As previously discussed, we analyze the scaling of the variance of a single term in~\eqref{eq:partialCsum} with fixed $x$ and $j$.  By invoking Lemma~\ref{lemma1}, $\Tr[G^x_j H_j]$ can be expressed as
\begin{align}
    \Tr[G^x_j H_j] = \sum_{\vec{p},\vec{q}}\Tr[V_j A^{(x,j)}_{\vec{q}\vec{p}} V_j\ad B^{(x,j)}_{\vec{p}\vec{q}}],\label{eq:gxj-hxj}
\end{align}
where 
\begin{align}
    A^{(x,j)}_{\vec{q}\vec{p}} & =\Tr_{\overline{j}}[(\id_{\tin, j}\otimes \ket{\vec{p}}\bra{\vec{q}}  )A^{(x,j)}] ,\quad A^{(x,j)} = (V_{j-1}\dots V_1)\sigma_x^{\text{in}}(V_1\ad\dots V_{j-1}\ad),\label{eq:Aqp-j}\\
    B^{(x,j)}_{\vec{p}\vec{q}} & = \Tr_{\overline{j}}[( \id_{\tin, j}\otimes \ket{\vec{q}}\bra{\vec{p}} )B^{(x,j)}],\quad B^{(x,j)} = [V_{j+1}\ad\ldots V_n\ad\sigma_x^{\text{out}} V_n \ldots V_{j+1},H_j]\,.\label{eq:Bpq-j}
\end{align}
Here $\Tr_{\overline{j}}$ indicates the trace over all qubits in the output layer except for qubit $j$. Moreover, we remark  that the summation in~\eqref{eq:gxj-hxj} runs over all bitstrings  $\vec{p}$, and $\vec{q}$ of lenght $n-1$, and we recall  that the operator $\ket{\vec{q}}\bra{\vec{p}}$  acts on all qubits in the output layer except on qubit $j$.
From the definition of $\sigma_x^\tin$ in~\eqref{eq:sigmas} and from \eqref{eq:Aqp-j}, it follows that $A^{(x,j)}_{\vec{q}\vec{p}}$ is nonzero when
\begin{align}
    p_k = q_k =0, \forall k \in \{j+1, \dots, n\}\,.\label{eq:pq-n}
\end{align}
Similarly, from $\sigma_x^\tout$ and~\eqref{eq:Bpq-j}, it follows that $B^{(x,j)}_{\vec{p}\vec{q}}$ is nonzero when
\begin{equation}
    p_k = q_k = z^{x}_k, \forall k \in \{1, \dots, j-1\}\,.  \label{eq:pq-1}
\end{equation}

Equations~\eqref{eq:pq-n} and \eqref{eq:pq-1} follow from the fact that each $V_j$ acts on $n+1$ number of qubits, and these equations imply that there is a single nonzero term in the summation~\eqref{eq:gxj-hxj}. This term can be identified by defining the following bitstring  of length $n-1$:
\begin{align}\label{eq:rxj}
\vec{r}^{(x, j)} \equiv (z^x_1,z^x_2,\ldots,z^x_{j-1},0,\ldots,0)\,.
\end{align} 
Then, by invoking Lemma~\eqref{eq:lemma3} we get  
\begin{align}
\langle (\Tr[G^x_j H_j])^2\rangle_{V_j}  &=    \int d\mu(V_j)  \Tr[V_j A^{(x,j)}_{\vec{r}^{(x, j)}\vec{r}^{(x, j)}} V_j\ad B^{(x,j)}_{\vec{r}^{(x, j)}\vec{r}^{(x, j)}}]\Tr[V_j A^{(x,j)}_{\vec{r}^{(x, j)} \vec{r}^{(x, j)}} V_j\ad B^{(x,j)}_{\vec{r}^{(x, j)}\vec{r}^{(x, j)}}] \\
&= \frac{1}{2^{2(n+1)}-1}\left(\Tr[( A^{(x,j)}_{\vec{r}^{(x, j)}\vec{r}^{(x, j)}})^2] - \frac{1}{2^{n+1}} \Tr[ A^{(x,j)}_{\vec{r}^{(x, j)}\vec{r}^{(x, j)}}]^2\right)\Tr[(B^{(x,j)}_{\vec{r}^{(x, j)}\vec{r}^{(x, j)}})^2]\\
& \leq \frac{1}{2^{2(n+1)}-1}\Tr[ (A^{(x,j)}_{\vec{r}^{(x, j)}\vec{r}^{(x, j)}})^2]\Tr[(B^{(x,j)}_{\vec{r}^{(x, j)}\vec{r}^{(x, j)}})^2]\,. \label{eq:gxj-hxj-vj}
\end{align}
where in the inequality we used the fact that $\Tr[ A^{(x,j)}_{\vec{r}^{(x, j)}\vec{r}^{(x, j)}}]>0$ as $A^{(x,j)}_{\vec{r}^{(x, j)}\vec{r}^{(x, j)}}$ is a positive semidefinite operator, and therefore, we can drop the term with the negative sign.

Since $A^{(x,j)}_{\vec{r}^{(x, j)}\vec{r}^{(x, j)}}$ and $B^{(x,j)}_{\vec{r}^{(x, j)}\vec{r}^{(x, j)}}$ are functions of different preceptrons $V_i$,  we can compute an upper bound on the expectation value of $(\Tr[G^x_j H_j])^2$ as follows:
\begin{equation}\label{eq:Abaverage}
    \langle\Tr[ (A^{(x,j)}_{\vec{r}^{(x, j)}\vec{r}^{(x, j)}})^2]\Tr[(B^{(x,j)}_{\vec{r}^{(x, j)}\vec{r}^{(x, j)}})^2]\rangle=\langle \Tr[(A^{(x,j)}_{\vec{r}^{(x, j)}\vec{r}^{(x, j)}})^2]\rangle_{V_1,\ldots V_{j-1}}\langle \Tr[(B^{(x,j)}_{\vec{r}^{(x, j)}\vec{r}^{(x, j)}})^2]\rangle_{V_{j+1},\ldots V_{n}}\,.
\end{equation}

We note that since $H_j$ only acts on all input qubits and on output qubit $j$, then   $B^{(x,j)}_{\vec{r}^{(x, j)}\vec{r}^{(x, j)}}$ can be expressed in the following compact form \begin{align}
B^{(x,j)}_{\vec{r}^{(x, j)}\vec{r}^{(x, j)}} = s_{\vec{r}^{(x, j)}}^{(x,j)} [\omega_{\vec{r}^{(x, j)}}^{(x,j)}, H_j],\label{eq:bppj}
\end{align}
where 
\begin{align}
\omega_{\vec{r}^{(x, j)}}^{(x,j)} &= \frac{1}{s_{\vec{r}^{(x, j)}}^{(x,j)}}\Tr_{\overline{j}}[( \id_{\tin, j}\otimes \ket{\vec{r}^{(x, j)}}\bra{\vec{r}^{(x, j)}} )V_{j+1}\ad\ldots V_n\ad\sigma_x^{\text{out}} V_n \ldots V_{j+1}],\label{eq:taux-j}\\
s_{\vec{r}^{(x, j)}}^{(x,j)} &= \Tr[( \id_{\tin, j}\otimes \ket{\vec{r}^{(x, j)}}\bra{\vec{r}^{(x, j)}} )V_{j+1}\ad\ldots V_n\ad\sigma_x^{\text{out}} V_n \ldots V_{j+1}].\label{eq:sx-j}
\end{align}
It is straightforward to note that $\omega_{\vec{r}^{(x, j)}}^{(x,j)}$ is a quantum state on all qubits in the input layer plus qubit $j$ in the output layer. Let us now consider the following chain of inequalities: 
\begin{align}
    \Tr[([\omega_{\vec{r}^{(x, j)}}^{(x,j)}, H_j])^2] 
    &= 2\left(\Tr[\omega_{\vec{r}^{(x, j)}}^{(x,j)}H_j\omega_{\vec{r}^{(x, j)}}^{(x,j)}H_j]-\Tr[(\omega_{\vec{r}^{(x, j)}}^{(x,j)})^2(H_j)^2]\right)\label{eq:ineqtauH-init}\\
    &\leq 2\Tr[\omega_{\vec{r}^{(x, j)}}^{(x,j)}H_j\omega_{\vec{r}^{(x, j)}}^{(x,j)}H_j] \nonumber\\
    & \leq 2\Tr[H_j\omega_{\vec{r}^{(x, j)}}^{(x,j)}H_j]\nonumber\\
    & =2\Tr[\omega_{\vec{r}^{(x, j)}}^{(x,j)}(H_j)^2] \nonumber\\
    & \leq 2\Tr[(H_j)^2]  \nonumber\\
    & \leq 2^{n+2}. \label{eq:ineqtauH}
\end{align}
The first inequality follows from the fact that  $\omega_{\vec{r}^{(x, j)}}^{(x,j)}(H_j)^2\omega_{\vec{r}^{(x, j)}}^{(x,j)}$ is a positive semidefinite operator, so that  $\Tr[(\omega_{\vec{r}^{(x, j)}}^{(x,j)})^2(H_j)^2]\geq 0 $. The second inequality follows by noting that  $\omega_{\vec{r}^{(x, j)}}^{(x,j)}\leq \id$ and $H_j\omega_{\vec{r}^{(x, j)}}^{(x,j)}H_j\geq 0$. The third inequality follows from the fact that $\omega_{\vec{r}^{(x, j)}}^{(x,j)}\leq \id$, and that $(H_j)^2\geq 0$. The last inequality holds from the assumption that $\Tr[(H_j)^2]\leq 2^{n+1}$. Finally, by combining Eqs.~\eqref{eq:bppj} and~\eqref{eq:ineqtauH}, we get that
\begin{align}
    \Tr[(B^{(j)}_{\vec{r}^{(x, j)}\vec{r}^{(x, j)}})^2] &= (s_{\vec{r}^{(x, j)}}^{(x,j)})^2\Tr[([\omega_{\vec{r}^{(x, j)}}^{(x,j)}, H_j])^2]\nonumber\\
    & \leq 2^{n+2} (s_{\vec{r}^{(x, j)}}^{(x,j)})^2.\label{eq:bound-bppj}
\end{align}

Let us now evaluate the term  $(s_{\vec{r}^{(x, j)}}^{(x,j)})^2$. By invoking Lemma~\ref{lemma1}, we get
\begin{align}\label{eq:sxj}
    s_{\vec{r}^{(x, j)}}^{(x,j)} &= \Tr[V_{j+1}( \id_{\tin, j}\otimes \ket{\vec{r}^{(x, j)}}\bra{\vec{r}^{(x, j)}})V_{j+1}\ad V_{j+2}\ad\cdots V_n\ad \sigma_x^{\tout}V_n\cdots V_{j+2}]\\
    & =\sum_{\vec{p}'\vec{q}'} \Tr[V_{j+1} C_{\vec{q}'\vec{p}'}^{(x,j+1)} V_{j+1}\ad D^{(x,j+1)}_{\vec{p}'\vec{q}'}].
\end{align}
Here the summation is over all  bitstrings $\vec{p}'$ and $\vec{q}'$  of length $n+1$, and 
\begin{align}
    C_{\vec{q}'\vec{p}'}^{(x,j+1)} & = \Tr_{\overline{j+1}}[( \id_{\tin, j+1}\otimes \ket{\vec{p}'}\bra{\vec{q}'})( \id_{\tin, j}\otimes \ket{\vec{r}^{(x, j)}}\bra{\vec{r}^{(x, j)}} )], \\
     D^{(x,j+1)}_{\vec{p}'\vec{q}'} & = \Tr_{\overline{j+1}}[( \id_{\tin, j+1}\otimes \ket{\vec{q}'}\bra{\vec{p}'})V_{j+2}\ad \cdots V_n\ad \sigma_x^{\tout}V_n\cdots V_{j+2}]\,,
\end{align}
where $\Tr_{\overline{j+1}}$ indicates the trace over all qubits in the output layer except qubit $j+1$.

Then from arguments similar to those used to deriving \eqref{eq:pq-n} and~\eqref{eq:pq-1}, we find 
\begin{align}
    q'_j &= p'_j=z^x_j, \\
    q'_k & = p'_k = r^{(x, j)}_k, \forall k \in \{1,2, \dots, j-1, j+2,\dots, n\}.
\end{align}
We now point at a recursive relation. Let \begin{align}
\vec{r}^{(x,j+1)}\equiv (r^{(x, j)}_1, \dots, r^{(x, j)}_{j-1},z^x_j, r^{(x, j)}_{j+2}, \dots r^{(x, j)}_{n}).
\end{align}
Then $s_{\vec{r}^{(x, j)}}^{(x,j)}$ further simplifies  to 
\begin{align}
    s_{\vec{r}^{(x, j)}}^{(x,j)} & =  \Tr\left[V_{j+1}(\id_{\tin}\otimes \ket{r^{(x,j)}_{j+1}}\bra{r^{(x,j)}_{j+1}})V_{j+1}\ad  \Tr_{\overline{j+1}}\left( \id_{\tin, j+1}\otimes (\ket{\vec{r}^{(x,j+1)}}\bra{\vec{r}^{(x,j+1)}}) V_{j+2}\ad\cdots V_n\ad \sigma_x^{\tout}V_n \cdots V_{j+2} \right)  \right]\\
    & = s_{\vec{r}^{(x,j+1)}}^{(x,j+1)}\Tr\left[V_{j+1}(\id_{\tin}\otimes \ket{r^{(x,j)}_{j+1}}\bra{r^{(x,j)}_{j+1}})V_{j+1}\ad  \omega_{\vec{r}^{(x,j+1)}}^{(x,j+1)} \right].\label{eq;recursion}
\end{align}
Here we defined
\begin{align}
    s_{\vec{r}^{(x, j+1)}}^{(x,j+1)} &= \Tr\left[( \id_{\tin, j+1}\otimes \ket{\vec{r}^{(x, j+1)}}\bra{\vec{r}^{(x, j+1)}} ) V_{j+2}\ad\cdots V_n\ad \sigma_x^{\tout}V_n \cdots V_{j+2} \right],\\
    \omega_{\vec{r}^{(x, j+1)}}^{(x,j+1)} & = \frac{1}{s_{\vec{r}^{(x, j+1)}}^{(x,j+1)}}\Tr_{\overline{j+1}}\left[( \id_{\tin, j+1}\otimes \ket{\vec{r}^{(x, j+1)}}\bra{\vec{r}^{(x, j+1)}} ) V_{j+2}\ad\cdots V_n\ad \sigma_x^{\tout}V_n \cdots V_{j+2} \right].  
\end{align}

An upper bound on the average of $(s_{\vec{r}^{(x,j)}}^{(x, j)})^2$ over $V_{j+1}$ can be obtained as follows, provided that $V_{j+1}$ forms a $2$-design:
\begin{align}
    &\langle (s_{\vec{r}^{(x,j)}}^{(x,j)})^2 \rangle_{V_{j+1}} \nonumber \\
    & = \int d\mu(V_{j+1}) (s_{\vec{r}^{(x,j)}}^{(x,j)})^2\\
    & = (s_{\vec{r}^{(x, j+1)}}^{(x,j+1)})^2 \int d\mu(V_{j+1}) \Tr\left[V_{j+1}(\id_{\tin}\otimes \ket{r^{(x,j)}_{j+1}}\bra{r^{(x,j)}_{j+1}})V_{j+1}\ad  \omega_{\vec{r}^{(x, j+1)}}^{(x,j+1)} \right]\Tr\left[V_{j+1}(\id_{\tin}\otimes \ket{r^{(x,j)}_{j+1}}\bra{r^{(x,j)}_{j+1}})V_{j+1}\ad  \omega_{\vec{r}^{(x, j+1)}}^{(x,j+1)} \right]\nonumber\\
    & = \frac{(s_{\vec{r}^{(x, j+1)}}^{(x,j+1)})^2}{2^{2(n+1)}-1}\left(2^{2n}+2^{n}\Tr[(\omega_{\vec{r}^{(x, j+1)}}^{(x,j+1)})^2]  - \frac{1}{2^{n+1}}(2^{n}+ 2^{2n}\Tr[(\omega_{\vec{r}^{(x, j+1)}}^{(x,j+1)})^2])  \right)\\
    & \leq \frac{2^n (2^n+1/2)(s_{\vec{r}^{(x, j+1)}}^{(x,j+1)})^2}{2^{2(n+1)}-1},\label{eq:sxj-avg}
\end{align}
where we employed~\eqref{eq:lemma3} and used the fact that $\Tr[(\omega_{\vec{r}^{(x, j+1)}}^{(x,j+1)})^2]\leq 1$ as $\omega_{\vec{r}^{(x, j+1)}}^{(x,j+1)}$ is a quantum state. 

Here we remark that from   $s_{\vec{r}^{(x, j+1)}}^{(x,j+1)}$ we can always define an operator $s_{\vec{r}^{(x, j+2)}}^{(x,j+2)}$ according to Eqs.~\eqref{eq:sxj}--\eqref{eq;recursion}. Moreover, by using the assumption that all randomly initialized perceptrons form $2$-designs, we can recursively average over $V_{j+2}, \dots, V_n$. Therefore, from \eqref{eq:bound-bppj}, we get  
\begin{align}
    \langle\Tr[(B^{(x,j)}_{\vec{r}^{(x, j)}\vec{r}^{(x, j)}})^2] \rangle_{V_{j+1},\ldots V_{n}} & \leq (s^{(x, n)}_{\vec{r}^{(x,n)}})^22^{n+1}\left(\frac{2^n (2^n+1/2)}{2^{2(n+1)}-1}\right)^{n-j}\\
& = 2^{3n+1}\left(\frac{2^n (2^n+1/2)}{2^{2(n+1)}-1}\right)^{n-j},    \label{eq:bj-avg}
\end{align}
where we used that fact that $s^{(x, n)}_{\vec{r}^{(x,n)}} = \Tr[\sigma_x^\tout]= 2^n$. 

We now compute the average of $\Tr[ (A^{(x,j)}_{\vec{r}^{(x, j)}\vec{r}^{(x, j)}})^2]$ over $V_1, \dots, V_{j-1}$. By following a similar procedure to the one previously employed, and from \eqref{eq:Aqp-j}, we have
\begin{align}
    A^{(x,j)}_{\vec{r}^{(x, j)}\vec{r}^{(x, j)}} & =\Tr_{\overline{j}}\left[( \id_{\tin, j}\otimes \ket{\vec{r}^{(x, j)}}\bra{\vec{r}^{(x, j)}} )(V_{j-1}\dots V_1)\sigma_x^{\text{in}}(V_1\ad\dots V_{j-1}\ad)\right]\\
    & = q_{\vec{r}^{(x, j)}}^{(x,j)} \varphi_{\vec{r}^{(x, j)}}^{(x,j)},\label{eq:appj}
\end{align}
where 
\begin{align}
    q_{\vec{r}^{(x, j)}}^{(x,j)} & =\Tr\left[( \id_{\tin, j}\otimes \ket{\vec{r}^{(x, j)}}\bra{\vec{r}^{(x, j)}} )(V_{j-1}\dots V_1)\sigma_x^{\text{in}}(V_1\ad\dots V_{j-1}\ad)\right] ,\\
    \varphi_{\vec{r}^{(x, j)}}^{(x,j)} & =\frac{1}{q_{\vec{r}^{(x, j)}}^{(x,j)}}\Tr_{\overline{j}}\left[( \id_{\tin, j}\otimes \ket{\vec{r}^{(x, j)}}\bra{\vec{r}^{(x, j)}} )(V_{j-1}\dots V_1)\sigma_x^{\text{in}}(V_1\ad\dots V_{j-1}\ad)\right] . 
\end{align}
Moreover, if $V_{j-1}$ forms a $2$-design, we can compute the expectation value of  $\Tr[ (A^{(x,j) }_{\vec{r}^{(x, j)}\vec{r}^{(x, j)}})^2]$ with respect $V_{j-1}$ as 
\begin{align}
\int d\mu(V_{j-1})\Tr[ (A^{(j) }_{\vec{r}^{(x, j)}\vec{r}^{(x, j)}})^2] 
&= \int d\mu(V_{j-1}) (q_{\vec{r}^{(x, j)}}^{(x,j)})^2 \Tr[(\varphi_{\vec{r}^{(x, j)}}^{(x,j)})^2]\\
& \leq \int d\mu(V_{j-1}) (q_{\vec{r}^{(x, j)}}^{(x,j)})^2\\
& = \int d\mu(V_{j-1}) \left(\Tr\left[V_{j-1}\ad( \id_{\tin, j}\otimes \ket{\vec{r}^{(x, j)}}\bra{\vec{r}^{(x, j)}} )V_{j-1}(V_{j-2}\dots V_1)\sigma_x^{\text{in}}(V_1\ad\dots V_{j-2}\ad)\right]\right)^2\nonumber \\
& \leq \frac{2^n(2^n+1/2)(q_{
\hat{\vec{r}}^{(x,j-1)}}^{(x,j-1)})^2}{2^{2(n+1)}-1},
\end{align}
where we used arguments similar to those used in deriving \eqref{eq:sxj}--\eqref{eq:sxj-avg}. Here, 
\begin{align}
    q_{\hat{\vec{r}}^{(x,j-1)}}^{(x, j-1)} &=\Tr\left[( \id_{\tin, j-1}\otimes\ket{\hat{\vec{r}}^{(x,j-1)}}\bra{\hat{\vec{r}}^{(x,j-1)}} )(V_{j-2}\dots V_1)\sigma_x^{\text{in}}(V_1\ad\dots V_{j-2}\ad)\right],  \\
    \hat{\vec{r}}^{(x,j-1)} & = (r^{(x,j)}_1, r^{(x,j)}_2, \dots, r^{(x,j)}_{j-2},0 , r^{(x,j)}_{j+1}, \dots, r^{(x,j)}_n),
\end{align}
where $ \hat{\vec{r}}^{(x,j-1)}$ denotes a bitstring of length $n-1$, and where $j-1$ in the superscript implies that $\ket{ \hat{\vec{r}}^{(x,j-1)}}$ is a state on all qubits in the output layer, except the $(j-1)$-th qubit. Then, since all randomly initialized perceptrons form $2$-designs, we can  recursively compute the average over $V_{j-2}, \dots V_1$. We get 
\begin{align}
\langle(\Tr[ (A^{(j) }_{\vec{r}^{(x, j)}\vec{r}^{(x, j)}})^2])^2 \rangle_{V_{1},\ldots V_{j-1}}  &\leq
   (q_{\hat{\vec{r}}^{(x,1)}}^{(x, 1)})^2\left(\frac{2^n (2^n+1/2) }{2^{2(n+1)}-1}\right)^{j-1}\\
   & = \left(\frac{2^n (2^n+1/2) }{2^{2(n+1)}-1}\right)^{j-1}, \label{eq:aj-avg}
\end{align}
where we used the fact that $q_{\hat{\vec{r}}^{(x,1)}}^{(x, 1)} = \Tr[\sigma_x^\tin]= 1$. 
Then from \eqref{eq:gxj-hxj-vj}, \eqref{eq:bj-avg}, and \eqref{eq:aj-avg}, it follows that 
\begin{align}
\langle (\Tr[G^x_j H_j])^2 \rangle &\leq   \frac{1}{2^{2(n+1)}-1}\langle\Tr[ (A^{(x,j)}_{\vec{r}^{(x, j)}\vec{r}^{(x, j)}})^2]\Tr[(B^{(x,j)}_{\vec{r}^{(x, j)}\vec{r}^{(x, j)}})^2]\rangle\\
& \leq \frac{2^{3n+2}}{{2^{2(n+1)}-1}} \left(\frac{2^n (2^n+1/2) }{2^{2(n+1)}-1}\right)^{n-1}\\
& \leq f(n) \in \mathcal{O}(1/2^n). 
\end{align}

\subsubsection{Fixed $j$ and different $x$}\label{sec:fixedjdiffx-prev}

We now establish an upper bound on the cross terms with equal $j$ but different $x$, i.e., on terms of the form: $\langle \Tr[G^x_j H_j]\Tr[G^{x'}_j H_j]\rangle$. Following \eqref{eq:gxj-hxj}--\eqref{eq:gxj-hxj-vj}, we find that 
\begin{align}
\langle \Tr[G^x_j H_j]\Tr[G^{x'}_j H_j]\rangle_{V_j} &\leq \frac{1}{2^{2(n+1)}-1}\left\vert \left(\Delta(A^{(x,x',j)})_{\vec{r}^{(x, j)}\vec{r}^{(x, j)}}^{\vec{r}^{(x', j)}\vec{r}^{(x', j)}}\right)\right\vert \left\vert\Tr[(B^{(x,j)}_{\vec{r}^{(x, j)}\vec{r}^{(x, j)}}B^{(x',j)}_{\vec{r}^{(x', j)}\vec{r}^{(x', j)}})]\right\vert. 
\end{align}
Then, from   \eqref{eq:bppj}--\eqref{eq:bound-bppj}, and by invoking the Cauchy-Schwarz inequality, we find that 
\begin{align}
\left|\Tr[(B^{(x,j)}_{\vec{r}^{(x, j)}\vec{r}^{(x, j)}}B^{(x',j)}_{\vec{r}^{(x', j)}\vec{r}^{(x', j)}})]\right| & \leq  \sqrt{\Tr[(B^{(x,j)}_{\vec{r}^{(x, j)}\vec{r}^{(x, j)}})^2]}\sqrt{\Tr[(B^{(x',j)}_{\vec{r}^{(x', j)}\vec{r}^{(x', j)}})^2]}\\
& \leq 2^{n+2} s_{\vec{r}^{(x, j)}}^{(x,j)}s^{(x',j)}_{\vec{r}^{(x', j)}}. 
\end{align}

From \eqref{eq:lemma3}, we compute the average of $s_{\vec{r}^{(x, j)}}^{(x,j)}s^{(x',j)}_{\vec{r}^{(x', j)}}$ with respect to $V_{j+1}$ as follows:
\begin{align}
     &\langle s_{\vec{r}^{(x, j)}}^{(x,j)}s^{(x',j)}_{\vec{r}^{(x', j)}} \rangle_{V_{j+1}} \nonumber \\ 
    & = \int d\mu(V_{j+1}) s_{\vec{r}^{(x, j)}}^{(x,j)}s^{(x',j)}_{\vec{r}^{(x', j)}} \label{eq:sx-sx'-vj-init}  \\
    & = s_{\vec{r}^{(x, j+1)}}^{(x,j+1)}s^{(x',j+1)}_{\vec{r}^{(x', j+1)}} \int d\mu(V_{j+1}) \Tr\left[V_{j+1}(\id_{\tin}\otimes \ket{r^{(x,j)}_{j+1}}\bra{r^{(x,j)}_{j+1}})V_{j+1}\ad  \omega_{\vec{r}^{(x,j+1)}}^{(x,j+1)} \right]\Tr\left[V_{j+1}(\id_{\tin}\otimes \ket{r^{(x',j)}_{j+1}}\bra{r^{(x',j)}_{j+1}})V_{j+1}\ad  \omega_{\vec{r}^{(x', j+1)}}^{(x',j+1)} \right]\nonumber\\
    & = \frac{s_{\vec{r}^{(x, j+1)}}^{(x,j+1)}s^{(x',j+1)}_{\vec{r}^{(x', j+1)}}}{2^{2(n+1)}-1}\left(2^{2n}+2^{n}\Tr[\omega_{\vec{r}^{(x,j+1)}}^{(x,j+1)}\omega_{\vec{r}^{(x',j+1)}}^{(x',j+1)}]  - \frac{1}{2^{n+1}}(2^{n}+ 2^{2n}\Tr[\omega_{\vec{r}^{(x,j+1)}}^{(x,j+1)}\omega_{\vec{r}^{(x',j+1)}}^{(x',j+1)}])  \right)\\
    & \leq \frac{2^n (2^n+1/2)s_{\vec{r}^{(x, j+1)}}^{(x,j+1)}s^{(x',j+1)}_{\vec{r}^{(x', j+1)}}}{2^{2(n+1)}-1}, \label{eq:sx-sx'-vj-final}
\end{align}
where we used the fact that $\Tr[\omega_{\vec{r}^{(x,j+1)}}^{(x,j+1)}\omega_{\vec{r}^{(x',j+1)}}^{(x',j+1)}]\leq 1$ for quantum states $\omega_{\vec{r}^{(x,j+1)}}^{(x,j+1)}$ and $\omega_{\vec{r}^{(x',j+1)}}^{(x',j+1)}$ defined as in~\eqref{eq:taux-j}. Then by recursively integrating over each perceptron we find 
\begin{align}
    \left\langle \left|\Tr[(B^{(x,j)}_{\vec{r}^{(x, j)}\vec{r}^{(x, j)}}B^{(x',j)}_{\vec{r}^{(x', j)}\vec{r}^{(x', j)}})]\right| \right\rangle_{V_{j+1},\ldots V_{n}} \leq 2^{3n+2}\left(\frac{2^n (2^n+1/2)}{2^{2(n+1)}-1}\right)^{n-j}\,.\label{eq:avg-bxj-bxprimej}
\end{align}

We now establish an upper bound on $\left| \left(\Delta(A^{(x,x',j)})_{\vec{r}^{(x, j)}\vec{r}^{(x, j)}}^{\vec{r}^{(x', j)}\vec{r}^{(x', j)}}\right)\right|$. Consider that 
\begin{align}
    \left| \left(\Delta(A^{(x,x',j)})_{\vec{r}^{(x, j)}\vec{r}^{(x, j)}}^{\vec{r}^{(x', j)}\vec{r}^{(x', j)}}\right)\right| &= \left| \Tr[ A^{(x,j)}_{\vec{r}^{(x, j)}\vec{r}^{(x, j)}} A^{(x',j)}_{\vec{r}^{(x', j)}\vec{r}^{(x', j)}}] - \frac{1}{2^{n+1}} \Tr[ A^{(x,j)}_{\vec{r}^{(x, j)}\vec{r}^{(x, j)}}]\Tr[ A^{(x',j)}_{\vec{r}^{(x', j)}\vec{r}^{(x', j)}}] \right|\\
    & = \left|q^{(x,j)}_{\vec{r}^{(x,j)}}q^{(x',j)}_{\vec{r}^{(x',j)}} \Tr[\varphi^{(x,j)}_{\vec{r}^{(x,j)}}\varphi^{(x',j)}_{\vec{r}^{(x',j)}}] - \frac{q^{(x,j)}_{\vec{r}^{(x,j)}}q^{(x',j)}_{\vec{r}^{(x',j)}}}{2^{n+1}} \Tr[\varphi^{(x,j)}_{\vec{r}^{(x,j)}}] \Tr[\varphi^{(x',j)}_{\vec{r}^{(x',j)}}] \right|\\
    & \leq q^{(x,j)}_{\vec{r}^{(x,j)}}q^{(x',j)}_{\vec{r}^{(x',j)}}\,. \label{eq:ax-xprime-j}
\end{align}
The second equality follows from \eqref{eq:appj}. The inequality follows from the fact that $\Tr[\varphi^{(x,j)}_{\vec{r}^{(x,j)}}\varphi^{(x',j)}_{\vec{r}^{(x',j)}}] \leq 1$ and $\Tr[\varphi^{(x,j)}_{\vec{r}^{(x,j)}}] =1$. Then by following arguments similar to \eqref{eq:sx-sx'-vj-init}--\eqref{eq:sx-sx'-vj-final} we find that 
\begin{align}
\left\langle \left|\left(\Delta(A^{(x,x',j)})_{\vec{r}^{(x, j)}\vec{r}^{(x, j)}}^{\vec{r}^{(x', j)}\vec{r}^{(x', j)}}\right)  \right| \right\rangle_{V_{1},\ldots V_{j-1}}  & \leq   \langle q^{(x,j)}_{\vec{r}^{(x,j)}}q^{(x',j)}_{\vec{r}^{(x',j)}} \rangle \\
&\leq q^{(x,1)}_{\vec{r}^{(x,1)}}q^{(x',1)}_{\vec{r}^{(x',1)}} \left(\frac{2^n (2^n+1/2)}{2^{2(n+1)}-1}\right)^{j-1}\\
& = \left(\frac{2^n (2^n+1/2)}{2^{2(n+1)}-1}\right)^{j-1}. \label{eq:aj-pp-bound-final}
\end{align}
Therefore, combining \eqref{eq:avg-bxj-bxprimej} and \eqref{eq:aj-pp-bound-final} leads to
\begin{align}\label{eq:gxprime-khk-thm1-sec-c}
    \langle  \Tr[G^x_j H_j]\Tr[G^{x'}_j H_j]\rangle &\leq \frac{2^{3n+2}}{2^{2(n+1)}-1}\left(\frac{2^n (2^n+1/2)}{2^{2(n+1)}-1}\right)^{n-1}\\
    & \leq f(n) \in \mathcal{O}(1/2^n)\,.
\end{align}

\subsubsection{Different $j$ and different $x$}\label{sec:diff-jx-thm1-gc}
In this subsection, we establish an upper bound on the average of cross terms of the form $\langle \Tr[G^x_j H_j]\Tr[G^{x'}_k H_k]\rangle$. Without loss of generality we assume that $j<k$. From~\eqref{eq:gxj-hxj-gen-form} we get
\begin{align}\label{eq:gxkhk}
        \Tr[G^{x'}_k H_k] = \Tr\left[V_j \cdots V_1 \sigma_x^{\tin}V_1\ad \cdots V_j\ad\cdot \left(V_{j+1}\ad \cdots V_k\ad   \left[V_{k+1}\ad\cdots V_n\ad \sigma_{x'}^{\tout}V_n \cdots V_{k+1}, H_k\right]V_{k}\cdots V_{j+1}\right) \right].
\end{align}

Then by following \eqref{eq:gxj-hxj}--\eqref{eq:gxj-hxj-vj}, we find that 
\begin{align}
\langle \Tr[G^x_j H_j]\Tr[G^{x'}_k H_k]\rangle_{V_j} 
 \leq \frac{1}{2^{2(n+1)}-1} \left(\Delta(A^{(x,x',j)})_{\vec{r}^{(x, j)}\vec{r}^{(x, j)}}^{\vec{r}^{(x', j)}\vec{r}^{(x', j)}}\right) \left(\Tr[(B^{(x,j)}_{\vec{r}^{(x, j)}\vec{r}^{(x, j)}}M^{(x',j)}_{\vec{r}^{(x', j)}\vec{r}^{(x', j)}})]\right),
\end{align}
where 
\begin{align}
   M^{(x',j)}_{\vec{r}^{(x', j)}\vec{r}^{(x', j)}} &= \Tr_{\overline{j}}[( \id_{\tin, j} \otimes \ket{\vec{r}^{(x', j)}}\bra{\vec{r}^{(x', j)}} )M^{(x',j)}],\label{eq:mrxj}\\
   M^{(x',j)} &=V_{j+1}\ad \cdots V_k\ad   \left[V_{k+1}\ad\cdots V_n\ad \sigma_{x'}^{\tout}V_n \cdots V_{k+1}, H_k\right]V_{k}\cdots V_{j+1}. \label{eq:mxj}
\end{align}
Moreover, from \eqref{eq:Bpq-j} it follows that 
\begin{align}\label{eq:brxj-sec-c}
    B^{(x,j)}_{\vec{r}^{(x, j)}\vec{r}^{(x, j)}} = \Tr_{\overline{j}}[( \id_{\tin, j}\otimes \ket{\vec{r}^{(x, j)}}\bra{\vec{r}^{(x, j)}} )[V_{j+1}\ad\ldots V_n\ad\sigma_x^{\text{out}} V_n \ldots V_{j+1},H_j]].
\end{align}

We now argue using Lemma \ref{lemma2} that the average of $\Tr[(B^{(x,j)}_{\vec{r}^{(x, j)}\vec{r}^{(x, j)}}M^{(x',j)}_{\vec{r}^{(x', j)}\vec{r}^{(x', j)}})]$ is zero. Here, $H_j$ in \eqref{eq:brxj-sec-c} and $H_k$ in \eqref{eq:mxj} correspond to $H$ and $K$ in Lemma \ref{lemma2}, respectively.
Moreover, $( \id_{\tin, j} \otimes \ket{\vec{r}^{(x, j)}}\bra{\vec{r}^{(x, j)}} )$ and $( \id_{\tin, j} \otimes \ket{\vec{r}^{(x', j)}}\bra{\vec{r}^{(x', j)}} )$ correspond to $P$ and $P'$, respectively. Furthermore, $V_{j+1}\ad$ corresponds to $V$, while $V_{j+2}\ad\ldots V_k\ad$ corresponds to $U$. Finally, $V_{k+1}\ad\ldots V_n\ad \sigma^{\tout}_{x} V_n \ldots V_{k+1}$ and   $V_{k+1}\ad\ldots V_n\ad \sigma^{\tout}_{x'} V_n \ldots V_{k+1}$  correspond to $S$ and $S'$, respectively in Lemma \ref{lemma2}. Hence from Lemma \ref{lemma2} it follows that 
\begin{align}
    \left\langle\left(\Tr[(B^{(x,j)}_{\vec{r}^{(x, j)}\vec{r}^{(x, j)}}M^{(x',j)}_{\vec{r}^{(x', j)}\vec{r}^{(x', j)}})]\right) \right\rangle_{V_{j+1}} = 0,
\end{align}
which implies  
\begin{align}
    \langle \Tr[G^x_j H_j]\Tr[G^{x'}_k H_k]\rangle \leq 0. 
\end{align}

Therefore, by combining results from Sections \ref{sec:fixed-jx-thm1-gc}--\ref{sec:diff-jx-thm1-gc}, it follows that
\begin{align}
    \langle (\partial_s C)^2 \rangle \leq f(n), 
\end{align}
 with $f(n)$ as in \eqref{eq:theo1eq}. 
 
 \subsection{Local Cost}
 We now estimate the scaling of the variance of the partial derivative of the local cost function. We first note that for a local cost function, \eqref{eq:partialCsum} gets transformed as follows: 
 \begin{align}\label{eq:costdoublesum}
    \partial_s C^L = \frac{i}{N}\sum_{x=1}^N\left[\frac{1}{n}\sum_{i=1}^n\bigg(\sum_{j=1}^n \Tr[G^x_{(i,j)} H_j]\bigg)\right],
\end{align}
where 
\begin{align}
\sigma_{(x,i)}^{\tout} &= \id_{\tin, \overline{i}}\otimes \dya{z^x_{i}},\label{eq:sigma-xi-out} \\
G^x_{(i,j)} &= [V_j \cdots V_1 \sigma_x^{\tin} V_1\ad \cdots V_j\ad, V_{j+1}\ad \cdots V_n\ad \sigma_{(x,i)}^{\tout}V_n \cdots V_{j+1}]. 
\end{align}

Moreover, from the cyclicity of trace, each term $\Tr[G^x_{(i,j)} H_j]$ can always be expressed as follows:
\begin{align}
    \Tr[G^x_{(i,j)} H_j] = \Tr\left[V_j \cdots V_1 \sigma_x^{\tin}V_1\ad \cdots V_j\ad \left[V_{j+1}\ad\cdots V_n\ad \sigma_{(x,i)}^{\tout}V_n \cdots V_{j+1}, H_j\right]\right].
\end{align}

We now provide a proof of Theorem \ref{theo1SM} for local cost functions in the following subsections. Similarly to the proof for the global cost in Section~\ref{sec:gc-thm1}, here we individually consider all different cases than can arise from the tripple summation in~\eqref{eq:costdoublesum}.

\subsubsection{Fixed j, fixed i, and fixed x}\label{sec:lc-sec-a-thm1}

Let us consider first the case $i <j$.  From \eqref{eq:Bpq-j} it follows that 
\begin{align}\label{eq:bxj-zero}
    B^{(x,i,j)} = [\id_{\tin, \overline{i}}\otimes \dya{z^x_i}, \id_{\overline{j}}\otimes H_j] = 0,
\end{align}
which can be combined with \eqref{eq:gxj-hxj} to imply that $\Tr[G^x_{i,j}H_j]=0$.

Let us consider the case when $i=j$. From \eqref{eq:Bpq-j} it follows that 
\begin{align}
B^{(x, i, j)}_{\vec{r}^{(x,j)}\vec{r}^{(x,j)}} = [\id_{\tin} \otimes \dya{z^x_j}, H_j], 
\end{align}
which further implies that 
\begin{align}\label{eq:bxj-id-zxj-hj}
    \Tr[(B^{(x,i, j)}_{\vec{r}^{(x,j)}\vec{r}^{(x,j)}})^2] \leq 2^{n+2}, 
\end{align}
where we used arguments similar to \eqref{eq:ineqtauH}. Then, combining the previous inequality with \eqref{eq:gxj-hxj-vj} and \eqref{eq:aj-avg}, we get 
\begin{align}
    \langle (\Tr[G^x_{i,j}, H_j])^2\rangle &   \leq \frac{2^{n+2}}{{2^{2(n+1)}-1}} \left(\frac{2^n (2^n+1/2) }{2^{2(n+1)}-1}\right)^{j-1}\\
& \leq f(n) \in \mathcal{O}(1/2^{n}). 
\end{align}

We now consider the case when $i>j$. By following \eqref{eq:Bpq-j} again, we get 
\begin{align}
    B^{(x,i,j)} &= [V_{j+1}\ad\ldots V_i\ad (\id_{\tin, j+1, \dots, i-1}\otimes \dya{z^x_i})V_i\ldots V_{j+1}\otimes \id_{j}, H_j]\otimes \id_{1, \dots,j-1,i+1,i+2, \dots, n},\\
    B^{(x, i,j)}_{\vec{r}^{(x,j)}} &= [\upsilon_{j}, H_j], 
\end{align}
where 
\begin{align}
    \upsilon_{j} &= \Tr_{\overline{\tin, j}}[(\id_{\tin,j}\otimes \dya{\vec{0}}_{j+1, \dots, i})\Upsilon],\\
    \Upsilon & = (V_{j+1}\ad\ldots V_i\ad (\id_{\tin,j, j+1, \dots, i-1}\otimes \dya{z^x_i})V_i\ldots V_{j+1}),
\end{align}
with $\upsilon_j$ acting on all input qubits and on the $j$-th output qubit, and where the subscript in  $\dya{\vec{0}}_{j+1, \dots, i}$ indicates that the projector acts on qubits $j+1, \dots, i$ in the output layer.

We note that $\Upsilon \leq \id_{\tin, j+1,\dots, i}$ since $(\id_{\tin,j, j+1, \dots, i-1}\otimes \dya{z^x_i}) \leq \id_{\tin,j, \dots i}$, and since unitary transformations do not change the spectrum of an operator. Then, the following inequality holds:
\begin{align}
  \Tr_{\overline{\tin, j}}[(\id_{\tin,j}\otimes \dya{\vec{0}}_{j+1, \dots, i})(\id_{\tin, j, \dots, i} - \Upsilon)(\id_{\tin,j}\otimes \dya{\vec{0}}_{j+1, \dots, i})] \geq 0,
\end{align}
which implies that $\upsilon_{j} \leq \id_{\tin, j}$. 
Using arguments similar to the ones used in deriving \eqref{eq:ineqtauH-init}--\eqref{eq:ineqtauH}, we find 
\begin{align}\label{eq:bxij-i-greater-j-thm1}
    \Tr[(B^{(x,i, j)}_{\vec{r}^{(x,j)}\vec{r}^{(x,j)}})^2] \leq 2^{n+2}. 
\end{align}
Again by combining this with \eqref{eq:gxj-hxj-vj} and \eqref{eq:aj-avg}, we get 
\begin{align}
    \langle (\Tr[G^x_{(i,j)}, H_j])^2\rangle &   \leq \frac{2^{n+2}}{{2^{2(n+1)}-1}} \left(\frac{2^n (2^n+1/2) }{2^{2(n+1)}-1}\right)^{j-1}\\
& \leq f(n) \in \mathcal{O}(1/2^{n}). 
\end{align}

\subsubsection{Fixed j, different i, and different x}

In this subsection, we establish a bound on $\langle (\Tr[G^x_{(i,j)}H_j]\Tr[G^{x'}_{(i',j)}H_j])\rangle$. We first note that if either of $i$ or $i'$ are smaller  than $j$, then from \eqref{eq:bxj-zero} we have 
\begin{align}
    \langle (\Tr[G^x_{(i,j)}H_j]\Tr[G^{x'}_{(i',j)}H_j)]\rangle = 0. 
\end{align}

Let $i = j$ and $i'>j$. By means of the Cauchy-Schwarz inequality and invoking both \eqref{eq:bxj-id-zxj-hj} and \eqref{eq:bxij-i-greater-j-thm1}, we find 
\begin{align}
\vert \Tr[B^{(x,i, j)}_{\vec{r}^{(x,j)}\vec{r}^{(x,j)}} B^{(x',i', j)}_{\vec{r}^{(x',j)}\vec{r}^{(x',j)}}]\vert &\leq \sqrt{\Tr[(B^{(x,i, j)}_{\vec{r}^{(x,j)}\vec{r}^{(x,j)}})^2]}\sqrt{(\Tr[B^{(x',i', j)}_{\vec{r}^{(x',j)}\vec{r}^{(x',j)}}])^2}\\
& =  2^{n+2}. 
\end{align}

Again by combining this with  \eqref{eq:aj-pp-bound-final}, we get 
\begin{align}
    \langle \Tr[G^x_{(i,j)}H_j] \Tr[G^{x'}_{(i',j)}H_j]\rangle 
& \leq f(n) \in \mathcal{O}(1/2^{n}). 
\end{align}

\subsubsection{Different j, different i, and different x}\label{sec:lc-sec-c-thm1}
In this subsection, we find an upper bound on $   \langle (\Tr[G^x_{(i,j)}H_j]\Tr[G^{x'}_{(i',k)}H_{k})]\rangle$. We first note that $\Tr[G^{x'}_{i',k} H_k]$ can be expressed as \eqref{eq:gxprime-khk-thm1-sec-c}, where $\sigma^{out}_{x'}$ is replaced by $    \sigma_{(x',i')}^{\tout}$ as in \eqref{eq:sigma-xi-out}. Without loss of generality we assume that $j<k$. 
Since the proof in Section \ref{sec:diff-jx-thm1-gc} holds for any form of $\sigma^{\tout}_x$ and $\sigma_{x'}^{\tout}$, it follows that 
\begin{align}
     \langle (\Tr[G^x_{(i,j)}H_j]\Tr[G^{x'}_{(i',k)}H_{k})]\rangle_{V_{j+1}} = 0,
\end{align}
and therefore, by combining results from Sections \ref{sec:lc-sec-a-thm1}--\ref{sec:lc-sec-c-thm1}, we find that
\begin{align}
\langle (\partial_s C^L)^2 \rangle \leq f(n),     
\end{align}
 with $f(n)$ as in \eqref{eq:theo1eq}.

\end{proof}

\setcounter{theorem}{0}

\section{Proof of Theorem \ref{theo2SM}}\label{sec:prooftheo2}

In this section, we provide a proof of Theorem \ref{theo2SM}, which we recall for covinience.  

\begin{theorem}\label{theo2SM}
Consider a DQNN with deep global perceptrons parametrized as in~\eqref{eq:parametrizationSM}, such that  $A_1^1$, $B_1^1$ in~\eqref{eq:AandB1}--~\eqref{eq:AandB2}, and $V_j^1$ ($\forall j$) form independent $2$-designs over $n+1$ qubits. Then, the variance of a partial derivative of the  cost function with respect to $\theta^\nu$ is upper bounded as
\begin{equation}
\Var[\partial_\nu C^G]\leq g(n), \quad \text{with} \quad g(n)\in\OC\left(1/2^{2n}\right)\,,
\end{equation}
if $O_x$ is the global operator of~\eqref{eq:globalopSM}, and upper bounded as 
\begin{equation}
\Var[\partial_\nu C^L]\leq h(n), \quad \text{with} \quad h(n)\in\OC\left(1/2^{n}\right)\,,
\end{equation}
when $O_x$ is the local operator in Eq.~\eqref{eq:localopSM}. 
\end{theorem}

\begin{proof}
Here we first analyze the global cost function, and then we consider the case of local cost functions. Similarly to the proofs of the previous sections, we divide our derivations in several subsections consisting of different cases.

\subsection{Global cost}
Similar to Section \ref{sec:proof-theo1}, we consider a DQNN with $n$ input and $n$ output qubits, and with no hidden layer.  Then, we recall that we can compute the partial derivative of the cost function with respect to a parameter $\theta^\nu$ in a given $V_j$, i.e., $\partial C/\partial \theta^\nu\equiv \partial_\nu C$,  as
\begin{align}\label{eq:csumx}
    \partial_\nu C= \frac{i}{2N}\sum_{x=1}^{N}\partial_\nu C_x\,, \quad \text{with} \quad \partial_\nu C_x \equiv \Tr\Big[A_j^1\widetilde{\sigma}^\tin_x(A_j^1)\ad [\id_{\overline{j}}\otimes \Gamma_k,(B_j^1)\ad\widetilde{\sigma}^\tout_xB_j^1]\Big]\,,
\end{align}
with $B_j^1=\id_{\overline{j}}\otimes \prod_{\nu=1}^{k-1}R_k(\theta^k)W_\nu$,  $A_j^1=\id_{\overline{j}}\otimes \prod_{\nu=k}^{\eta_j}R_k(\theta^k)W_\nu$, and where 
\begin{align}
    \widetilde{\sigma}^\tin_x=V_{j-1}\ldots V_1 \sigma^\tin_x V_1\ad\ldots V_{j-1}\ad\,,\label{eq:tilde-sigma-xin-thm2}\\
    \widetilde{\sigma}^\tout_x=V_{j+1}\ad\ldots V_{n}\ad \sigma^\tout_x V_{n}\ldots V_{j+1}\,.\label{eq:tilde-sigma-xout-thm2}
\end{align}

\subsubsection{Fixed x}

We now establish an upper bound on a single term in~\eqref{eq:csumx}. That is, we consider a term  $\langle (\partial_\nu C_x)^2 \rangle$ with fixed $x$. From \eqref{eq:lemma3} it follows that 
\begin{align}\label{eq:delk-cx}
    \langle (\partial_\nu C_x)^2 \rangle_{A_j^1, B_j^1} =\frac{2^{n}\Tr[\Gamma_k^2]}{(2^{2n+2}-1)^2}\sum_{\substack{\vec{p}\vec{q}\\\vec{p}'\vec{q}'}}\Delta(\Omega^{(x,j)})_{\vec{q}\vec{p}}^{\vec{q}'\vec{p}'}\Delta(\Psi^{(x,j)})_{\vec{p}\vec{q}}^{\vec{p}'\vec{q}'}\,,
\end{align}    
where the summation runs over all bitstrings $\vec{p}$, $\vec{q}$, $\vec{p}'$, $\vec{q}'$ of length $n-1$. In addition, we defined 
\begin{align}
        \Delta(\Omega^{(x,j)})_{\vec{q}\vec{p}}^{\vec{q}'\vec{p}'} &= \Tr[ \Omega_{\vec{q}\vec{p}}^{(x,j)}\Omega_{\vec{q}'\vec{p}'}^{(x,j)}] -\frac{\Tr[ \Omega_{\vec{q}\vec{p}}^{(x,j)}]\Tr[\Omega_{\vec{q}'\vec{p}'}^{(x,j)}]}{2^{n+1}},\label{eq:del-psi-cx-thm2}\\
    \Delta(\Psi^{(x,j)})_{\vec{p}\vec{q}}^{\vec{p}'\vec{q}'} &=  \Tr[\Psi_{\vec{p}\vec{q}}^{(x,j)}\Psi_{\vec{p}'\vec{q}'}^{(x,j)}]-\frac{\Tr[\Psi_{\vec{p}\vec{q}}^{(x,j)}]\Tr[\Psi_{\vec{p}'\vec{q}'}^{(x,j)}]}{2^{n+1}}\,,\label{eq:del-k-cx-thm2}
\end{align}
 where  $\Omega_{\vec{q}\vec{p}}^{(x,j)}$ and $\Psi_{\vec{q}\vec{p}}^{(x,j)}$ are operators on $n+1$ qubits (all qubits in the input layer plus the $j$-th qubit in the output layer) defined as 

\begin{align}
    \Omega_{\vec{q}\vec{p}}^{(x,j)}&=\Tr_{\overline{j}}\left[(\id_{\tin,j}\otimes \ket{\vec{p}}\bra{\vec{q}}_{\overline{j}}) \widetilde{\sigma}^\tout_x \right]\label{eq:omega-qp}\\
    \Psi_{\vec{p}\vec{q}}^{(x,j)}&=\Tr_{\overline{j}}\left[(\id_{\tin,j}\otimes \ket{\vec{q}}\bra{\vec{p}}_{\overline{j}})\widetilde{\sigma}^\tin_x\right]\,,\label{eq:psi-pq},
\end{align} 
and where $\Tr_{\overline{j}}$ indicates the trace over  the subsystem of  all qubits in the output layer except for the $j$-th qubit. 

Similar to Section \ref{sec:proof-theo1},  we assume that the output state is in the computational basis  $\ket{\phi_x^{\tout}}\equiv \ket{\vec{z}^x}=\ket{z^x_1z^x_2\ldots z^x_n}$. Then  following arguments similar to the ones employed in Eq.~\eqref{eq:pq-n} and \eqref{eq:pq-1}, we find that
\begin{align}
    p_k &= q_k = 0, \forall k \in \{j+1, \dots, n\},\label{eq:pqk-1}\\
    p_k &= q_k = z_k^x, \forall k \in \{1, \dots, j-1\}\,, \label{eq:pqk-n}
\end{align}
which leads to a bitstring  $\vec{r}^{(x,j)}$ as in \eqref{eq:rxj}.

By recursively integrating over each randomly initialized perceptron as in  \eqref{eq:appj}--\eqref{eq:aj-avg}, we find that 
\begin{align}
\langle \Delta(\Psi^{(x,j)})_{\vec{r}^{(x,j)}\vec{r}^{(x,j)}}^{\vec{r}^{(x,j)}\vec{r}^{(x,j)}}\rangle_{V_1,\ldots,V_{j-1}} & \leq \langle \Tr[(\Psi^{(x,j)}_{\vec{r}^{(x,j)}\vec{r}^{(x,j)}})^2] \rangle \\
& \leq \left(\frac{2^n (2^n+1/2) }{2^{2(n+1)}-1}\right)^{j-1} \,.
\end{align}

On the other hand, using \eqref{eq:taux-j} and \eqref{eq:sx-j}, the following inequality holds:
\begin{align}
    \Delta(\Omega^{(x,j)})_{\vec{r}^{(x, j)}\vec{r}^{(x, j)}}^{\vec{r}^{(x, j)}\vec{r}^{(x, j)}}  &\leq (s_{\vec{r}^{(x, j)}}^{(x,j)})^2 \Tr[(\omega_{\vec{r}^{(x, j)}}^{(x,j)})^2]\\
    & \leq (s_{\vec{r}^{(x, j)}}^{(x,j)})^2\label{eq:new-m1}
\end{align}
where $s_{\vec{r}^{(x, j)}}^{(x,j)}$ is given by \eqref{eq:sx-j}, and where we again used the fact that $\Tr[(\omega_{\vec{r}^{(x, j)}}^{(x,j)})^2]\leq 1$ as $\omega_{\vec{r}^{(x, j)}}^{(x,j)}$ is a quantum state. Finally,  following \eqref{eq:sxj}--\eqref{eq:bj-avg}, we find  
\begin{align}
\langle  (s_{\vec{r}^{(x, j)}}^{(x,j)})^2\rangle_{V_{j+1},\ldots,V_n} \leq  2^{2n}\left(\frac{2^n (2^n+1/2)}{2^{2(n+1)}-1}\right)^{n-j}.\label{eq:srxj-thm2}
\end{align}

Therefore, combining \eqref{eq:psi-pq}, \eqref{eq:srxj-thm2}, and \eqref{eq:del-k-cx-thm2}, leads to
\begin{align}
     \langle (\partial_\nu C_x)^2 \rangle &\leq \frac{2^{3n}\Tr[\Gamma_k^2]}{(2^{2n+2}-1)^2}\left(\frac{2^n (2^n+1/2)}{2^{2(n+1)}-1}\right)^{n-1}\\
     & \leq \frac{2^{4n+1}}{(2^{2n+2}-1)^2}\left(\frac{2^n (2^n+1/2)}{2^{2(n+1)}-1}\right)^{n-1}\\
     & \leq g(n) \in \mathcal{O}(1/2^{2n}).
\end{align}

\subsubsection{Different x}
We now establish a bound on cross terms with different $x$, i.e., on terms of the form $\langle \partial_\nu C_x \partial_\nu C_{x'}  \rangle $.  From \eqref{eq:lemma3}, it follows that 
\begin{align}\label{eq:delk-cx-delk-cxprime-a1b1}
        \langle \partial_\nu C_x\partial_\nu C_{x'} \rangle_{A_j^1, B_j^1} =\frac{2^{n}\Tr[\Gamma_k^2]}{(2^{2n+2}-1)^2}\sum_{\substack{\vec{p}\vec{q}\\\vec{p}'\vec{q}'}}\Delta(\Omega^{(x,x',j)})_{\vec{q}\vec{p}}^{\vec{q}'\vec{p}'}\Delta(\Psi^{(x,x',j)})_{\vec{p}\vec{q}}^{\vec{p}'\vec{q}'},
\end{align}
where 
\begin{align}
        \Delta(\Omega^{(x,x',j)})_{\vec{q}\vec{p}}^{\vec{q}'\vec{p}'} &= \Tr[ \Omega_{\vec{q}\vec{p}}^{(x,j)}\Omega_{\vec{q}'\vec{p}'}^{(x',j)}] -\frac{\Tr[ \Omega_{\vec{q}\vec{p}}^{(x,j)}]\Tr[\Omega_{\vec{q}'\vec{p}'}^{(x',j)}]}{2^{n+1}},\label{eq:del-k-cx-thm2211}\\
    \Delta(\Psi^{(x,x,'j)})_{\vec{p}\vec{q}}^{\vec{p}'\vec{q}'} &=  \Tr[\Psi_{\vec{p}\vec{q}}^{(x,j)}\Psi_{\vec{p}'\vec{q}'}^{(x',j)}]-\frac{\Tr[\Psi_{\vec{p}\vec{q}}^{(x,j)}]\Tr[\Psi_{\vec{p}'\vec{q}'}^{(x',j)}]}{2^{n+1}}\,,\label{eq:del-k-cx-thm22}
\end{align}
and where $\Omega_{\vec{q}\vec{p}}^{(x,j)}$, and $\Psi_{\vec{q}\vec{p}}^{(x,j)}$ are defined according to Eqs.~\eqref{eq:omega-qp}, and ~\eqref{eq:psi-pq}, respectively.

From arguments similar to those used in deriving \eqref{eq:ax-xprime-j}--\eqref{eq:aj-pp-bound-final}, we find that
\begin{align}
 \left\langle    \Delta(\Psi^{(x,x',j)})_{\vec{r}^{(x,j)}\vec{r}^{(x,j)}}^{\vec{r}^{(x',j)}\vec{r}^{(x',j)}} \right\rangle_{V_1,\ldots,V_{j-1}}  \leq\left(\frac{2^n (2^n+1/2)}{2^{2(n+1)}-1}\right)^{j-1}\,.
\end{align}

Similarly, from arguments similar to \eqref{eq:bppj}--\eqref{eq:sx-j}, we obtain 
\begin{align}
\Omega^{(x, j)}_{\vec{r}^{(x, j)}\vec{r}^{(x,j)}} = s^{(x,j)}_{\vec{r}^{(x,j)}} \omega^{(x,j)}_{\vec{r}^{(x,j)}}, 
\end{align}
where $\omega^{(x,j)}_{\vec{r}^{(x,j)}}$ and $s^{(x,j)}_{\vec{r}^{(x,j)}}$  are given by \eqref{eq:sx-j} and \eqref{eq:taux-j}, respectively. Then, it is straightforward to show that 
\begin{align}
  \left\langle \Delta(\Omega^{(x,x',j)})_{\vec{r}^{(x, j)}\vec{r}^{(x, j)}}^{\vec{r}^{(x', j)}\vec{r}^{(x', j)}} \right\rangle_{V_{j+1},\ldots,V_n}  &\leq \langle s_{\vec{r}^{(x, j)}}^{(x,j)} s^{(x',j)}_{\vec{r}^{(x'{\tiny }, j)}}\rangle_{V_{j+1},\ldots,V_n} \\
    & \leq 2^{2n}\left(\frac{2^n (2^n+1/2)}{2^{2(n+1)}-1}\right)^{n-j}. \label{eq:delta-omega-x-xprime-j-thm2}
\end{align}

Therefore, by combing \eqref{eq:delk-cx-delk-cxprime-a1b1}--\eqref{eq:delta-omega-x-xprime-j-thm2}, we find 
\begin{align}
\langle \partial_\nu C_x \partial_\nu C_{x'} \rangle &\leq \frac{2^{4n+1}}{(2^{2n+2}-1)^2}    \left(\frac{2^n (2^n+1/2)}{2^{2(n+1)}-1}\right)^{n-1}\\
& \leq g(n) \in \mathcal{O}(1/2^{2n}). 
\end{align}
Hence, recalling that  $\partial_\nu C = (i/2N)\sum_{x=1}^N \partial_\nu C_x$, we get  
\begin{align}
    \left\langle (\partial_\nu C)^2 \right\rangle \leq g(n) \in \mathcal{O}\left(\frac{1}{2^{2n}}\right). 
\end{align}

\subsection{Local Cost}
 We now estimate the scaling of the variance of the partial derivative of the local cost function. We define the local cost function as follows:
\begin{align}\label{eq:sumlocalSM}
    C^L&=\frac{1}{nN}\sum_{x=1}^N \sum_{i=1}^{n}C^L_{x,i},\quad \text{with}\quad 
   C_{x,i}^L= \Tr\left[\sigma^{\tout}_{(x,i)} V_n\ldots V_1\sigma^\tin_x V_1\ad\ldots U_n\ad\right]\,,
\end{align}
and where $\sigma^{\tout}_{(x,i)}=\id_{\tin,\overline{i}}\otimes\dya{z^x_i}$. Here $\id_{\tin,\overline{i}}$ indicates the identity over all qubits in the input layer plus all the qubits in the output layer except for qubit $i$.

In what follows we first consider a single term in the summation over $x$ in~\eqref{eq:sumlocalSM}. Here we have to consider the three following cases $i<j$, $i>j$, and $i=j$. Moreover, we remark that  \eqref{eq:tilde-sigma-xin-thm2}--\eqref{eq:psi-pq} remain the same, except for the fact that $\sigma_x^{\tout}$ is replaced by $\sigma^{\tout}_{(x,i)}$.

\subsubsection{Fixed  $x$ and  $i<j$.}\label{sec:case1-thm2}

We first consider the case $i < j$. From \eqref{eq:omega-qp} and \eqref{eq:psi-pq}, we find that 
\begin{align}
    p_k &= q_k = 0, \forall k \in \{j+1, \dots, n\},\\
    p_i &= q_i = z^x_i,\\
    p_k &= q_k, \forall k \in \{1, \dots, i-1, i+1, \dots j-1\}. 
\end{align}
Then, we define the following set of bitstrings of length $n-1$: 
\begin{align}
    \vec{r}^{(x, i, j, \vec{p})} = (p^{(x, i, j)}_1, p^{(x, i, j)}_2, \dots, p^{(x, i, j)}_{i-1}, z^x_i, p^{(x, i, j)}_{i+1}, \dots p^{(x, i, j)}_{j-1}, 0, \dots, 0),\label{eq:rxij}
\end{align}
where $j$ in the superscript implies that $ \vec{r}^{(x, i, j,\vec{p})}$ is a bitstring over all qubits in the output layer, except the $j$-th qubit. The bold notation $\vec{p}$ in~\eqref{eq:rxij} indicates that  each $p^{(x, i, j)}_k \in \{0, 1\}$.  Then from \eqref{eq:omega-qp}, we find that
\begin{align}\label{eq:omega-rxijp-thm2-c1}
\Omega^{(x,i,j,\vec{p})}_{\vec{r}^{(x,i,j,\vec{p})}\vec{r}^{(x,i,j,\vec{p})}}&=\Tr_{\overline{j}}\left[(\id_{\tin,j}\otimes \ket{\vec{r}^{(x,i,j,\vec{p})}}\bra{\vec{r}^{(x,i,j,\vec{p})}}) (\id_{\tin,\overline{i}}\otimes\dya{z^x_i})\right]\\
& = \id_{\tin, j},
\end{align}
which further implies  
\begin{align}
  \Delta(\Omega^{(x,i,j,\vec{p}, \vec{p}')})_{\vec{r}^{(x,i,j,\vec{p})}\vec{r}^{(x,i,j,\vec{p})}}^{\vec{r}^{(x,i,j,\vec{p}')}\vec{r}^{(x,i,j,\vec{p}')}} &= \Tr[\id_{\tin, j}\id_{\tin, j}]-\frac{\Tr[\id_{\tin, j}]\Tr[\id_{\tin, j}]}{2^{n+1}}\\
  & = 2^{n+1}-\frac{2^{2(n+1)}}{2^{n+1}}\\
  & = 0. \label{eq:delta-omega-rxijp-thm2-c1}
\end{align}
Therefore, from \eqref{eq:delk-cx} we get
\begin{align}
        \langle (\partial_\nu C_x)^2 \rangle_{A_j^1, B_j^1}  = 0.   
\end{align}

\subsubsection{Fixed  $x$ and  $i>j$.}\label{sec:case2-thm2}
Let us consider the case when $i>j$. We note that this case is different from the one studied in the previous section due to the fact that the perceptron unitaries do not commute with each other. We now have 
\begin{align}
    p_k &= q_k = 0, \forall k \in \{j+1, \dots, n\},\\
    p_k &= q_k, \forall k \in \{1, \dots, j-1\}. 
\end{align}
Similarly to \eqref{eq:rxij}, we define  here the bitstrings
\begin{align}\label{eq:rbxijp}
    \vec{\hat{r}}^{(x, i, j, \vec{p})} = (p^{(x, i, j)}_1, p^{(x, i, j)}_2, \dots, p^{(x, i, j)}_{j-1},  0, \dots, 0).
\end{align}

In this case, from~\eqref{eq:omega-qp} we obtain the operator
\begin{align}\label{eq:omega-xij-secb-lc-thm2}
        \Omega^{(x,i,j)} &= \Tr_{j+1,\dots,i}[(  \id_{\tin} \otimes \ket{\vec{0}}\bra{\vec{0}}_{j+1, \dots i})V_{j+1}\ad\ldots V_i\ad (\id_{\tin,j+1, \dots, i-1}\otimes\ket{z^x_i}\bra{z^x_i}) V_i\ldots V_{j+1}]\otimes \id_j
\end{align}
where $\id_j$ is the identity over qubit $j$ in the output layer, and where $\ket{\vec{0}}\bra{\vec{0}}_{j+1, \dots i}$ is the projector onto the all-zero state on qubits $j+1,\ldots,i$ in the output layer. Note that now $\Omega^{(x,i,j)}$ in \eqref{eq:omega-xij-secb-lc-thm2}, and $\Delta(\Omega^{(x,i,j)})$ in~\eqref{eq:del-psi-cx-thm2} are independent of the  bitstring $\vec{p}$. 

Similarly, by using \eqref{eq:del-k-cx-thm2} we define 
\begin{align}
    \Psi^{(x,i,j,\vec{p})}_{\vec{\hat{r}}^{(x, i, j, \vec{p})}\vec{\hat{r}}^{(x, i, j, \vec{p})}} &= \Tr_{\overline{j}}[(\id_{\tin,j}\otimes \ket{\vec{\hat{r}}^{(x, i, j, \vec{p})}}\bra{\vec{\hat{r}}^{(x, i, j, \vec{p})}})V_{j-1}\ldots V_{1} (\sigma_x^{\tin}) V_1\ad\ldots V_{j-1}\ad]\label{eq:proof1}\\
    & = q^{(x,i,j,\vec{p})}_{\vec{\hat{r}}^{(x, i, j, \vec{p})}\vec{\hat{r}}^{(x, i, j, \vec{p})}}\varphi^{(x, i, j, \vec{p})}_{\vec{\hat{r}}^{(x, i, j, \vec{p})}\vec{\hat{r}}^{(x, i, j, \vec{p})}},
\end{align}
where 
\begin{align}
    q^{(x,i,j,\vec{p})}_{\vec{\hat{r}}^{(x, i, j, \vec{p})}\vec{\hat{r}}^{(x, i, j, \vec{p})}} & =  \Tr[(\id_{\tin,j}\otimes \ket{\vec{\hat{r}}^{(x, i, j, \vec{p})}}\bra{\vec{\hat{r}}^{(x, i, j, \vec{p})}})V_{j-1}\ldots V_{1} (\sigma_x^{\tin}) V_1\ad\ldots V_{j-1}\ad] ,\\
    \varphi^{(x, i, j, \vec{p})}_{\vec{\hat{r}}^{(x, i, j, \vec{p})}\vec{\hat{r}}^{(x, i, j, \vec{p})}} & = \frac{1}{q^{(x,i,j,\vec{p})}_{\vec{\hat{r}}^{(x, i, j, \vec{p})}\vec{\hat{r}}^{(x, i, j, \vec{p})}}} \Tr_{\overline{j}}[(\id_{\tin,j}\otimes \ket{\vec{\hat{r}}^{(x, i, j, \vec{p})}}\bra{\vec{\hat{r}}^{(x, i, j, \vec{p})}})V_{j-1}\ldots V_{1} (\sigma_x^{\tin}) V_1\ad\ldots V_{j-1}\ad].
\end{align}

We now note that 
\begin{align}
    \Delta( \Psi^{(x,i,j,\vec{p})})_{\vec{\hat{r}}^{(x, i, j, \vec{p})}\vec{\hat{r}}^{(x, i, j, \vec{p})}}^{\vec{\hat{r}}^{(x, i, j, \vec{p})}\vec{\hat{r}}^{(x, i, j, \vec{p})}}  \leq  q^{(x,i,j,\vec{p})}_{\vec{\hat{r}}^{(x, i, j, \vec{p})}\vec{\hat{r}}^{(x, i, j, \vec{p})}} q^{(x,i,j,\vec{p}')}_{\vec{\hat{r}}^{(x, i, j, \vec{p}')}\vec{\hat{r}}^{(x, i, j, \vec{p}')}},\label{eq:delta-psi-xijp-lc-secb-thm2}
\end{align}
which follows from the fact that $\Tr[(\varphi^{(x, i, j, \vec{p})}_{\vec{\hat{r}}^{(x, i, j, \vec{p})}\vec{\hat{r}}^{(x, i, j, \vec{p})}})^2] \leq 1$ as $\varphi^{(x, i, j, \vec{p})}_{\vec{\hat{r}}^{(x, i, j, \vec{p})}\vec{\hat{r}}^{(x, i, j, \vec{p})}}$ is a quantum state. 

Then by combining \eqref{eq:omega-xij-secb-lc-thm2} and \eqref{eq:delta-psi-xijp-lc-secb-thm2}, we get the following inequality: 
\begin{align}
       \langle (\partial_\nu C^L_{x,i})^2 \rangle_{A_j^1, B_j^1}&\leq \frac{2^{2n+1}}{(2^{2n+2}-1)^2} \Tr[(\Omega^{(x,i,j)})^2]\sum_{\vec{p}, \vec{p}'}q^{(x,i,j,\vec{p})}_{\vec{\hat{r}}^{(x, i, j, \vec{p})}\vec{\hat{r}}^{(x, i, j, \vec{p})}} q^{(x,i,j,\vec{p}')}_{\vec{\hat{r}}^{(x, i, j, \vec{p}')}\vec{\hat{r}}^{(x, i, j, \vec{p}')}}\\
       &= \frac{2^{2n+1}}{(2^{2n+2}-1)^2} \Tr[(\Omega^{(x,i,j)})^2]\left(\sum_{\vec{p}}q^{(x,i,j,\vec{p})}_{\vec{\hat{r}}^{(x, i, j, \vec{p})}\vec{\hat{r}}^{(x, i, j, \vec{p})}}\right)\left(\sum_{\vec{p}'} q^{(x,i,j,\vec{p}')}_{\vec{\hat{r}}^{(x, i, j, \vec{p}')}\vec{\hat{r}}^{(x, i, j, \vec{p}')}}\right)\\
       & = \frac{2^{2n+1}}{(2^{2n+2}-1)^2} \Tr[(\Omega^{(x,i,j)})^2], 
\end{align}
where we used the fact that 
\begin{align}\label{eq:qxijp-1-thm2-c2}
    \sum_{\vec{p}}q^{(x,i,j,\vec{p})}_{\vec{\hat{r}}^{(x, i, j, \vec{p})}\vec{\hat{r}}^{(x, i, j, \vec{p})}} =\Tr[(\id_{\tin,1,\ldots,j-1}\otimes \ket{\vec{0}}\bra{\vec{0}}_{j+1,\ldots,n})V_{j-1}\ldots V_{1} (\sigma_x^{\tin}) V_1\ad\ldots V_{j-1}\ad] = 1. 
\end{align}

Finally, by recursively applying Lemma \ref{lemma3}, we find that 
\begin{align}
\langle (\partial_\nu C^L_{x,i})^2 \rangle
 & \leq \frac{2^{2n+1}}{(2^{2n+2}-1)^2}\left\langle \Tr[(\Omega^{(x,i,j)})^2] \right\rangle_{V_{j+1},\ldots,V_{i}} \label{eq:d56} \\
     & \leq h(n) \in \mathcal{O}(1/2^{n})\,.
\end{align}

\subsubsection{Fixed  $x$ and  $i=j$.}\label{sec:case3-thm2}

In this case it can be easily shown that 
\begin{align}\label{eq:omega-xij-thm2-c3}
     \Omega^{(x,j,j)} = \id_{\tin} \otimes \dya{z^x_j},
\end{align}
and hence from~\eqref{eq:del-psi-cx-thm2} we have
\begin{equation}
    \Delta(\Omega^{(x,j)})\leq 2^n.
\end{equation}
Then, following a similar procedure to the one employed in the previous section (see Eqs.\eqref{eq:proof1}--\eqref{eq:qxijp-1-thm2-c2}), we obtain
\begin{align}
     \langle (\partial_\nu C^L_{x,i})^2 \rangle &\leq \frac{2^{3n+1}}{(2^{2n+2}-1)^2}\\
     & \leq h(n) \in \mathcal{O}(1/2^{n}). 
\end{align}

\subsubsection{Different  $x$ and different $i$.}
In this section we consider the cross terms of the form: $\langle \partial_\nu C^L_{x,i}\partial_\nu C^L_{x',i'}\rangle$. We first consider the case when either of $i$ or $i'$ is smaller than $j$. Then from arguments similar to those used in deriving \eqref{eq:omega-rxijp-thm2-c1}--\eqref{eq:delta-omega-rxijp-thm2-c1}, it can be shown that 
\begin{align}
      \Delta(\Omega^{(x,x',i,i',j,\vec{p}, \vec{p}')})_{\vec{r}^{(x,i,j,\vec{p})}\vec{r}^{(x,i,j,\vec{p})}}^{\vec{r}^{(x',i',j,\vec{p}')}\vec{r}^{(x',i',j,\vec{p}')}} = 0,
\end{align}
and therefore, 
\begin{align}
\langle (\partial_\nu C^L_{x,i}\partial_\nu C^L_{x',i'})   \rangle  = 0 . 
\end{align}

Let us now consider the case when $i=j$ and $i'>j$. Then from \eqref{eq:omega-xij-thm2-c3} and \eqref{eq:omega-xij-secb-lc-thm2} it follows that 
\begin{align}
     \Omega^{(x,i,j)} &= \id_{\tin} \otimes \dya{z^x_j},\\
             \Omega^{(x',i',j)} &= \Tr_{j+1,\dots,i'}[(  \id_{\tin} \otimes \ket{\vec{0}}\bra{\vec{0}}_{j+1, \dots i'})V_{j+1}\ad\ldots V_{i'}\ad (\id_{\tin,j+1, \dots, i'-1}\otimes\ket{z^x_{i'}}\bra{z^x_{i'}}) V_{i'}\ldots V_{j+1}]\otimes \id_j,
\end{align}
which further implies that $\Delta(\Omega^{(x,x',i,i',j)})$ in~\eqref{eq:del-psi-cx-thm2} is independent of the bitstrings $\vec{p}$ and $\vec{p}'$, and hence
\begin{align}
     \Delta(\Omega^{(x,x',i,i',j)}) &\leq \Tr[\widetilde{\Omega}^{(x',i',j)}]\Tr[ \dya{z^x_j}]\\
     & = \Tr[\widetilde{\Omega}^{(x',i',j)}],
\end{align}
where 
\begin{align}
\widetilde{\Omega}^{(x',i',j)} = \Tr_{j+1,\dots,i'}[(  \id_{\tin} \otimes \ket{\vec{0}}\bra{\vec{0}}_{j+1, \dots i'})V_{j+1}\ad\ldots V_{i'}\ad (\id_{\tin,j+1, \dots, i'-1}\otimes\ket{z^x_{i'}}\bra{z^x_{i'}}) V_{i'}\ldots V_{j+1}]. 
\end{align}

Finally, by using  arguments similar to those used in deriving \eqref{eq:qxijp-1-thm2-c2}, and by recursively invoking Lemma~\ref{lemma1}, the following bound holds:
\begin{align}
\langle (\partial_\nu C^L_{x,i}\partial_\nu C^L_{x',i'})   \rangle  &\leq \frac{2^{3n+1}}{(2^{2n+2}-1)^2}\label{eq:d69} \\
     & \leq h(n) \in \mathcal{O}(1/2^{n}). 
\end{align}

Let us finally consider the final case when both $i$ and $i'$ are greater than $j$ and $i<i'$. By following the proof in \ref{sec:case2-thm2}, we find that 
\begin{align}
    \langle (\partial_\nu C^L_{x,i}\partial_k  &C^L_{x',i'})   \rangle_{A_j^1, B_j^1}\nonumber \\
   &\leq \frac{2^{2n+1}}{(2^{2n+2}-1)^2} \left(\Tr[\Omega^{(x,i,j)}\Omega^{(x',i',j)}]\right)\left(\sum_{\vec{p}}q^{(x,i,j,\vec{p})}_{\vec{\hat{r}}^{(x, i, j, \vec{p})}\vec{\hat{r}}^{(x, i, j, \vec{p})}}\right)\left(\sum_{\vec{p}'} q^{(x',i',j,\vec{p}')}_{\vec{\hat{r}}^{(x', i', j, \vec{p}')}\vec{\hat{r}}^{(x', i', j, \vec{p}')}}\right)\\
       & = \frac{2^{2n+1}}{(2^{2n+2}-1)^2} \left(\Tr[\Omega^{(x,i,j)}\Omega^{(x',i',j)}]\right).
\end{align}

Then by following arguments as the one used in deriving \eqref{eq:d56} and \eqref{eq:d69}, we get
\begin{align}
    \langle (\partial_  k C^L_{x,i}\partial_\nu C^L_{x',i'})   \rangle 
     & \leq h(n) \in \mathcal{O}(1/2^{n}).
\end{align}
Therefore, by combining results from Sections \ref{sec:case1-thm2}--\ref{sec:case3-thm2}, it follows that 
\begin{align}
    \langle (\partial_\nu C^L)^2 \rangle \leq h(n) \in \mathcal{O}(1/2^n).
\end{align}
\end{proof}

\section{DQNNs with unitaries acting on $n+m$ qubits}\label{sec:dqnnnm}

In this section, we generalize our results for the case when unitaries in a DQNN acts on $n+m$ qubits. Here, $n$ denotes the number of qubits in the layer $l$ and $m$ qubits are from the layer $l+1$. In Sections \ref{sec:proof-theo1} and \ref{sec:prooftheo2}, we proved our results for the case when $m=1$. For instance, in Fig.~\ref{fig:new1}, we show a DQNN with two layers where the initial layer has 2 nodes and the final layer has 4 nodes. Here, $n=m=2$, and the action of perceptrons can be described in terms of two unitaries, each acting on four qubits (2 in the input layer, and 2 in the output layer). Note that, we henceforth assume that each qubit in the $(l+1)$-th layers is acted upon by only one perceptron. 

Since our theorem statements for this case are generalizable from the existing proofs, we only provide a brief sketch of our proof. As discussed previously, Theorem~1 and Theorem~2 for two different methods of updating the perceptrons. Below, we guide readers to a derivation similar to that of Section~\ref{sec:gc-thm1} and then state a new theorem for the general case. 

Let us consider the case when the cost function is defined in terms of the global operator in \eqref{eq:globalopSM}. We consider a particular example where the output state is on $mn$ qubits. In general, it does not have to depend on $n$. 

Following the arguments similar to \eqref{eq:gxj-hxj}--\eqref{eq:Bpq-j}, we get 
\begin{align}\label{eq:pkqk-sectionE}
    p_k &= q_k = 0, \forall k \in \{mj+1, \dots, mn\}~, 
    p_k = q_k = z_k^x, \forall k \in \{1, \dots, m(j-1)\}~.
\end{align}

To write \eqref{eq:pkqk-sectionE} compactly, we define the following bitstring of length $m(n-1)$: 
\begin{align}
    \vec{r}^{(x, j)} \equiv (z_1^x, z_2^x, \dots, z^x_{m(j-1)}, 0, \dots, 0)~.
\end{align}

\begin{figure}[t]
    \centering
    \includegraphics[width=.7\columnwidth]{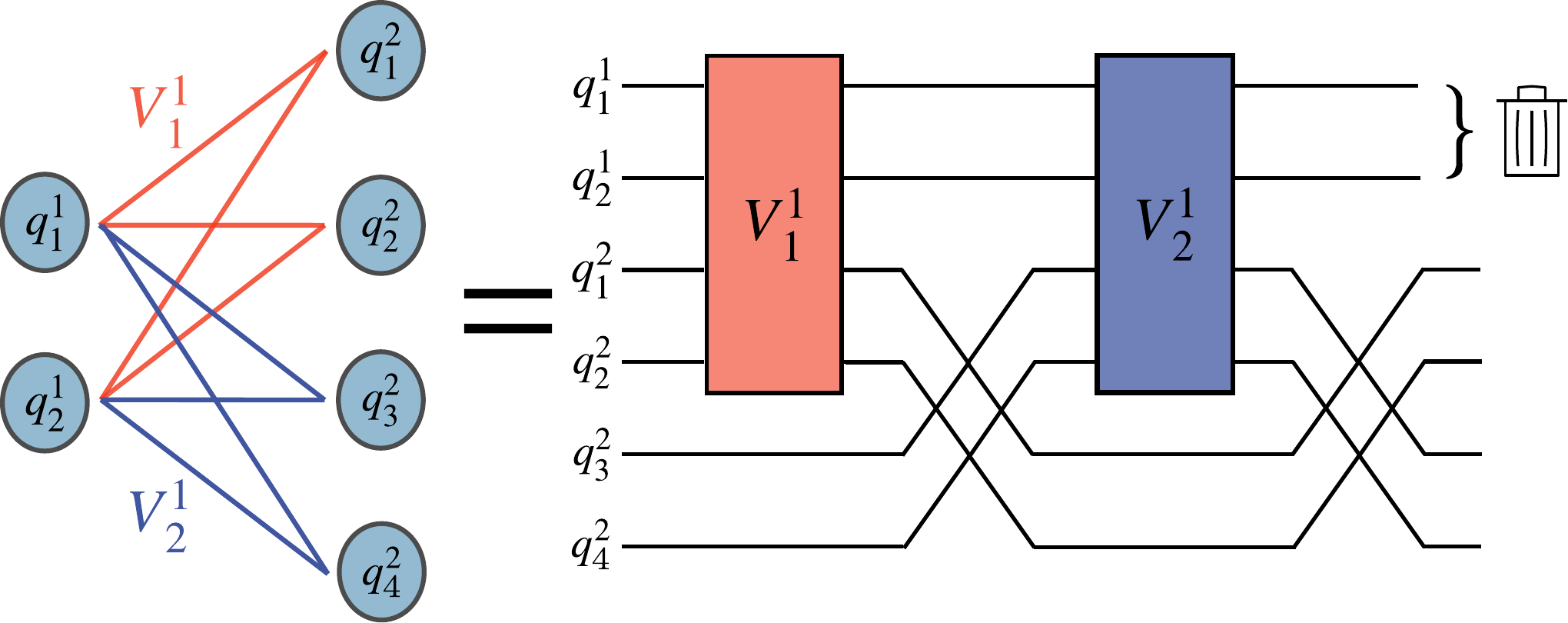}
    \caption{Schematic diagram of a DQNN where each perceptron acts non trivially on $n$ qubits in the $l$-th layer, and $m$ qubits in the $(l+1)$-th layer. Shown is the case when $n=m=2$.}
    \label{fig:new1}
\end{figure}

Using $\vec{r}^{(x,j)}$ we define $A^{(x,j)}_{\vec{r}^{(x, j)}\vec{r}^{(x, j)}}$ and $B^{(x,j)}_{\vec{r}^{(x, j)}\vec{r}^{(x, j)}}$, as in \eqref{eq:Aqp-j} and \eqref{eq:Bpq-j}, respectively. Then by invoking Lemma \ref{eq:lemma3} we get 
\begin{align}
\langle (\Tr[G^x_j H_j])^2\rangle_{V_j}  &=    \int d\mu(V_j)  \Tr[V_j A^{(x,j)}_{\vec{r}^{(x, j)}\vec{r}^{(x, j)}} V_j\ad B^{(x,j)}_{\vec{r}^{(x, j)}\vec{r}^{(x, j)}}]\Tr[V_j A^{(x,j)}_{\vec{r}^{(x, j)} \vec{r}^{(x, j)}} V_j\ad B^{(x,j)}_{\vec{r}^{(x, j)}\vec{r}^{(x, j)}}] \\
&= \frac{1}{2^{2(n+m)}-1}\left(\Tr[( A^{(x,j)}_{\vec{r}^{(x, j)}\vec{r}^{(x, j)}})^2] - \frac{1}{2^{n+m}} \Tr[ A^{(x,j)}_{\vec{r}^{(x, j)}\vec{r}^{(x, j)}}]^2\right)\Tr[(B^{(x,j)}_{\vec{r}^{(x, j)}\vec{r}^{(x, j)}})^2]\\
& \leq \frac{1}{2^{2(n+m)}-1}\Tr[ (A^{(x,j)}_{\vec{r}^{(x, j)}\vec{r}^{(x, j)}})^2]\Tr[(B^{(x,j)}_{\vec{r}^{(x, j)}\vec{r}^{(x, j)}})^2]\,, 
\end{align}
where for getting the inequality we used the fact that $\Tr[ A^{(x,j)}_{\vec{r}^{(x, j)}\vec{r}^{(x, j)}}]>0$ as $A^{(x,j)}_{\vec{r}^{(x, j)}\vec{r}^{(x, j)}}$ is a positive semidefinite operator.

We follow \eqref{eq:Abaverage} to compute the upper bound on the expectation value of $(\Tr[G^x_j H_j])^2$. Let us note that since $H_j$ only acts on all input qubits and on output qubits $m(j-1)+1, \dots mj$, i.e., $m$ output qubits in total, then   $B^{(x,j)}_{\vec{r}^{(x, j)}\vec{r}^{(x, j)}}$ can be expressed in the following compact form \begin{align}
B^{(x,j)}_{\vec{r}^{(x, j)}\vec{r}^{(x, j)}} = s_{\vec{r}^{(x, j)}}^{(x,j)} [\omega_{\vec{r}^{(x, j)}}^{(x,j)}, H_j],\label{eq:bppj}
\end{align}
where 
\begin{align}
\omega_{\vec{r}^{(x, j)}}^{(x,j)} &= \frac{1}{s_{\vec{r}^{(x, j)}}^{(x,j)}}\Tr_{\overline{j}}[( \id_{\tin, j}\otimes \ket{\vec{r}^{(x, j)}}\bra{\vec{r}^{(x, j)}} )V_{j+1}\ad\ldots V_n\ad\sigma_x^{\text{out}} V_n \ldots V_{j+1}],\\
s_{\vec{r}^{(x, j)}}^{(x,j)} &= \Tr[( \id_{\tin, j}\otimes \ket{\vec{r}^{(x, j)}}\bra{\vec{r}^{(x, j)}} )V_{j+1}\ad\ldots V_n\ad\sigma_x^{\text{out}} V_n \ldots V_{j+1}].
\end{align}

Here, the subscript $\overline{j}$ implies all output qubits besides $\{m(j-1)+1, m(j-1)+2, \dots, mj\}$. Similarly, $\id_{\tin, j}$ implies an identity acting on all qubits in the input layer and on qubits $\{m(j-1)+1, m(j-1)+2, \dots, mj\}$ in the output layer. We henceforth follow this notation.   
Moreover, note that $\omega_{\vec{r}^{(x, j)}}^{(x,j)}$ is a quantum state on all qubits in the input layer plus $j$-th block of $m$ qubits (i.e., $\{m(j-1)+1, m(j-1)+2, \dots, mj\}$ qubits) in the output layer. Then from arguments similar to those used in deriving \eqref{eq:ineqtauH-init}--\eqref{eq:ineqtauH}, we get 
\begin{align}
      \Tr[([\omega_{\vec{r}^{(x, j)}}^{(x,j)}, H_j])^2] \leq 2^{n+m+1}~,\label{eq:ineqtauH-new}
\end{align} 
where we used the assumption that $\Tr[(H_j)^2]\leq 2^{n+m}$. 

Finally, by combining Eqs.~\eqref{eq:bppj} and~\eqref{eq:ineqtauH-new}, we get that
\begin{align}
    \Tr[(B^{(j)}_{\vec{r}^{(x, j)}\vec{r}^{(x, j)}})^2] &= (s_{\vec{r}^{(x, j)}}^{(x,j)})^2\Tr[([\omega_{\vec{r}^{(x, j)}}^{(x,j)}, H_j])^2]\nonumber\\
    & \leq 2^{n+m+1} (s_{\vec{r}^{(x, j)}}^{(x,j)})^2.\label{eq:bound-bppj-new}
\end{align}

Let us now evaluate the term  $(s_{\vec{r}^{(x, j)}}^{(x,j)})^2$. By invoking Lemma~\ref{lemma1}, we get
\begin{align}\label{eq:sxj-new}
    s_{\vec{r}^{(x, j)}}^{(x,j)} &= \Tr[V_{j+1}( \id_{\tin, j}\otimes \ket{\vec{r}^{(x, j)}}\bra{\vec{r}^{(x, j)}})V_{j+1}\ad V_{j+2}\ad\cdots V_n\ad \sigma_x^{\tout}V_n\cdots V_{j+2}]\\
    & =\sum_{\vec{p}'\vec{q}'} \Tr[V_{j+1} C_{\vec{q}'\vec{p}'}^{(x,j+1)} V_{j+1}\ad D^{(x,j+1)}_{\vec{p}'\vec{q}'}].
\end{align}
Here the summation is over all  bitstrings $\vec{p}'$ and $\vec{q}'$  of length $m(n-1)$, and 
\begin{align}
    C_{\vec{q}'\vec{p}'}^{(x,j+1)} & = \Tr_{\overline{j+1}}[( \id_{\tin, j+1}\otimes \ket{\vec{p}'}\bra{\vec{q}'})( \id_{\tin, j}\otimes \ket{\vec{r}^{(x, j)}}\bra{\vec{r}^{(x, j)}} )], \\
     D^{(x,j+1)}_{\vec{p}'\vec{q}'} & = \Tr_{\overline{j+1}}[( \id_{\tin, j+1}\otimes \ket{\vec{q}'}\bra{\vec{p}'})V_{j+2}\ad \cdots V_n\ad \sigma_x^{\tout}V_n\cdots V_{j+2}]\,,
\end{align}
where $\Tr_{\overline{j+1}}$ indicates the trace over all qubits in the output layer except qubits $\{mj+1,\dots, m(j+1)\}$.

Then from arguments similar to those used to deriving \eqref{eq:pq-n} and~\eqref{eq:pq-1}, we find 
\begin{align}
    q'_k &= p'_k=z^x_k, \forall k \in \{m(j-1)+1, \dots, mj\}\\
    q'_k & = p'_k = r^{(x, j)}_k, \forall k \in \{1,2, \dots, m(j-1), m(j+1)+1,\dots, mn\}.
\end{align}
Let 
\begin{align}
\vec{r}^{(x,j+1)}\equiv (r^{(x, j)}_1, \dots, r^{(x, j)}_{m(j-1)},z^x_{m(j-1)+1},\dots, z^x_{mj},r^{(x, j)}_{m(j+1)+1}, \dots r^{(x, j)}_{n})~.
\end{align}
Then the average of $s_{\vec{r}^{(x, j)}}^{(x,j)}$ over $V_{j+1}$ can be upper bounded as follows: 
\begin{align}
    s_{\vec{r}^{(x, j)}}^{(x,j)} 
     \leq \frac{2^n (2^n+1/2)(s_{\vec{r}^{(x, j+1)}}^{(x,j+1)})^2}{2^{2(n+m)}-1},\label{eq:sxj-avg-new}
\end{align}
where we employed~\eqref{eq:lemma3} and used arguments similar to those used in deriving \eqref{eq:sxj-avg}. 

Here we remark that from   $s_{\vec{r}^{(x, j+1)}}^{(x,j+1)}$ we can always define an operator $s_{\vec{r}^{(x, j+2)}}^{(x,j+2)}$ according to Eqs.~\eqref{eq:sxj}--\eqref{eq;recursion}. Moreover, by using the assumption that all randomly initialized perceptrons form $2$-designs, we can recursively average over $V_{j+2}, \dots, V_n$. Therefore, from \eqref{eq:bound-bppj-new}, we get  
\begin{align}
    \langle\Tr[(B^{(x,j)}_{\vec{r}^{(x, j)}\vec{r}^{(x, j)}})^2] \rangle_{V_{j+1},\ldots V_{n}} & \leq (s^{(x, n)}_{\vec{r}^{(x,n)}})^22^{n+m+1}\left(\frac{2^n (2^n+1/2)}{2^{2(n+m)}-1}\right)^{n-j}\\
& = 2^{3n+m+1}\left(\frac{2^n (2^n+1/2)}{2^{2(n+m)}-1}\right)^{n-j},    
\end{align}
where we used that fact that $s^{(x, n)}_{\vec{r}^{(x,n)}} = \Tr[\sigma_x^\tout]= 2^{n}$. 

We now compute the average of $\Tr[ (A^{(x,j)}_{\vec{r}^{(x, j)}\vec{r}^{(x, j)}})^2]$ over $V_1, \dots, V_{j-1}$. By following a similar procedure to the one previously employed, and we get 
\begin{align}
\langle(\Tr[ (A^{(j) }_{\vec{r}^{(x, j)}\vec{r}^{(x, j)}})^2])^2 \rangle_{V_{1},\ldots V_{j-1}}  \leq \left(\frac{2^n (2^n+1/2) }{2^{2(n+m)}-1}\right)^{j-1}~. \label{eq:aj-avg-new}
\end{align}

Then from \eqref{eq:gxj-hxj-vj}, \eqref{eq:bj-avg}, and \eqref{eq:aj-avg}, it follows that 
\begin{align}
\langle (\Tr[G^x_j H_j])^2 \rangle
& \leq \frac{2^{3n+m+1}}{{2^{2(n+m)}-1}} \left(\frac{2^n (2^n+1/2) }{2^{2(n+m)}-1}\right)^{n-1}\\
& \leq f(n,m) \in \mathcal{O}(1/2^{(2m-1)n}).
\end{align}

Thus, using the aforementioned result, and from a similar analysis of Sections \ref{sec:fixedjdiffx-prev}--\ref{sec:lc-sec-c-thm1}, we find that for both global cost functions
\begin{align}
    \langle (\partial_s C^G)^2 \rangle  \leq f(n,m) \in \mathcal{O}(1/2^{(2m-1)n})~.
\end{align}

Similarly, one con generalize the proof for local cost functions.

\section{DQNNs with hidden layers}\label{sec:dqnn-hidden}

\begin{figure}[t]
    \centering
    \includegraphics[width=.4\columnwidth]{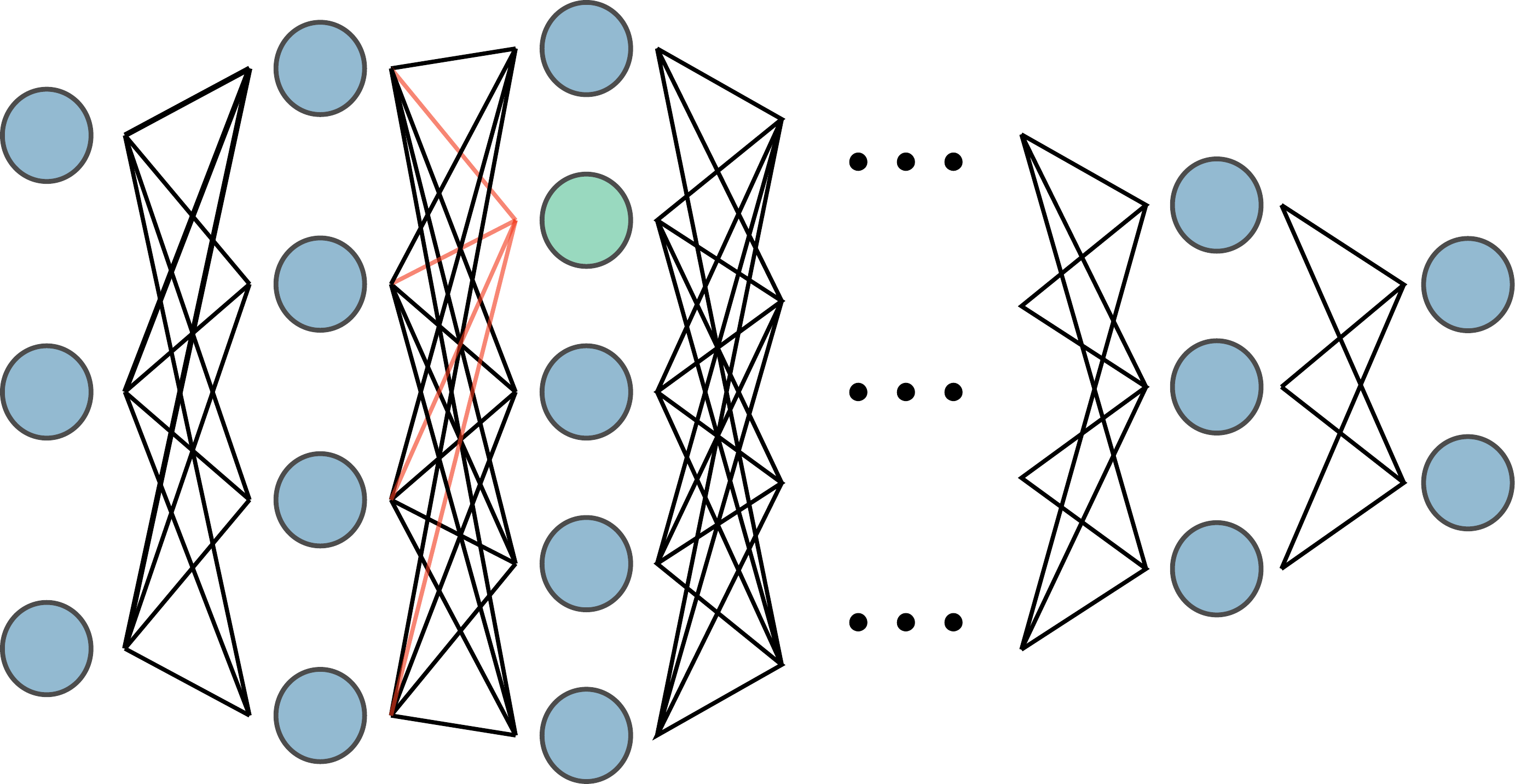}
    \caption{Schematic diagram for generalizing the results to a DQNN with an arbitrary number of layers. We now consider the case when the partial derivative is taken with respect to a parameter in $V_j^l$, i.e., in a perceptron of the $l$-th layer acting on the $j$-th qubit of the $(l+1)$-th layer.  }
    \label{fig:new}
\end{figure}

In this section we show how the results obtained in Sections \ref{sec:proof-theo1} and \ref{sec:prooftheo2} can be generalized to the case when the DQNN has $L$ hidden layers. 

First, let us follow the proof in Section~\ref{sec:prooftheo2} for a single input state labeled $x$, and note that here we can redefine the following quantities from Eq.~\eqref{eq:tilde-sigma-xin-thm2}:
\begin{align}
    \sigma_x^{\tin} &= \ketbra{\phi_x^{\tin}}{\phi_x^{\tin}}\otimes\ketbra{\vec{0}}{\vec{0}}_{\thid}\otimes  \ketbra{\vec{0}}{\vec{0}}_{\text{out}}\,,\label{eq:sigmas2} \\
    \sigma_x^{\tout}  &=  \id_{\tin}\otimes \id_{\thid}\otimes O_x\,. \label{eq:sigmas3}
\end{align}
Here, $\ket{\vec{0}}_{\thid}$ denotes the input state to the hidden layers, and $\id_{\thid}$ denotes the identity in the Hilbert space associated with the qubits in the hidden layers.

As shown in Fig.~\eqref{fig:new}, we now consider the case when the partial derivative is taken with respect to a parameter in $V_j^l$, i.e., in a perceptron of the $l$-th layer acting on the $j$-th qubit of the $(l+1)$-th layer. For convenience of notation, let us here define define the following sets of qubit indexes:
\begin{itemize}
    \item Set $\SC_1$: composed of all the indexes for qubits  in the input layer, and all those for qubits in the first $(l-2)$ hidden layers.
    \item Set $\SC_2$: composed of all the indexes for qubits  in the  $(l-1)$-th hidden layers.
    \item Set $\SC_3$: composed of all the indexes for qubits in the $l$-th layer with indexes smaller than j.
    \item Set $\SC_4$: composed of all the indexes for qubits in the $l$-th layer with indexes larger than j, and all those for the qubits in remaining hidden layers with indexes larger than $l$.
\end{itemize}
Therefore, the union of these sets along with the index of $j$-th qubit in the $l$-th layer describe all the qubits in the DQNN. 

Here we note that the action of the unitaries prior to $V_j^l$, make it so that the quantum state in Eq.~\eqref{eq:sigmas2} can be expressed as
\begin{align}
    \sigma_x^{(i,j-)} &= \dya{\phi_x^{(i,j)}}_{\SC_1,\SC_2}\otimes \ketbra{\vec{0}}{\vec{0}}_{j,\SC_3} \,,\label{eq:sigmasn1} 
\end{align}
where $\ket{\phi_x^{(i,j)}}_{\SC_1,\SC_2}$ is the joint state of all the qubits with indexes in the sets $\SC_1$ and $\SC_2$.

From the definition in~\eqref{eq:omega-qp} and from \eqref{eq:psi-pq}, it follows that
\begin{align}
    &p_k = q_k =0,  \quad \text{ for all qubits with index in $\SC_4$} \,.\label{eq:pq-n-n1}\\
    &p_k = q_k = z^{x}_k,   \quad \text{ for all qubits with index in $\SC_2$ and $\SC_3$}\,,  \label{eq:pq-1-n2}\\
    &p_k = q_k\,,  \quad \text{ for all qubits with index in $\SC_1$}\,.  \label{eq:pq-1-n3}
\end{align}
Then, let us recall the definition $\Omega_{\vec{q}\vec{p}}^{(x,j)}=\Tr_{\SC_1,\SC_3,\SC_4}\left[(\id_{\SC_2,j}\otimes \ket{\vec{p}}\bra{\vec{q}}_{\SC_1,\SC_3,\SC_4}) \widetilde{\sigma}^\tout_x \right]$. Here, since $\widetilde{\sigma}^\tout_x$ acts trivially on all qubits in $\SC_1$, we have that 
\begin{equation}
    \Tr_{\SC_1,\SC_3,\SC_4}\left[(\id_{\SC_2,j}\otimes \ket{\vec{p}}\bra{\vec{q}}_{\SC_1,\SC_3,\SC_4}) \widetilde{\sigma}^\tout_x \right]=\Tr_{\SC_3,\SC_4}\left[(\id_{\SC_2,j}\otimes \ket{\vec{p}}\bra{\vec{q}}_{\SC_3,\SC_4}) (\id_{\SC_2,j,\SC_3,\SC_4} \otimes O_x) \right]\,.
\end{equation}
and hence $\Omega_{\vec{q}\vec{p}}^{(x,j)}$ is independent of the $\vec{q}$ and $\vec{p}$ indexes in $\SC_1$. 

From Eq.~\eqref{eq:delk-cx}, we can now take the summation over $\vec{q}$ and $\vec{p}$ indexes in $\SC_1$ to note that
\begin{align}\label{eq:delk-cx}
   \sum_{\vec{p}\vec{q}\in\SC_1}\Psi_{\vec{q}\vec{p}}^{(x,j)}=\Tr_{\SC_1,\SC_3,\SC_4}\left[(\id_{\SC_2,j}\otimes \ket{\vec{p}}\bra{\vec{q}}_{\SC_1,\SC_3,\SC_4})\widetilde{\sigma}^\tin_x\right]=\Tr_{\SC_1,\SC_3,\SC_4}\left[(\id_{\SC_1,\SC_2,j}\otimes \ket{\vec{p}}\bra{\vec{q}}_{\SC_3,\SC_4})\widetilde{\sigma}^\tin_x\right]\,,
\end{align}   
where we have used the fact that $\sum_{\vec{p}\vec{q}\in\SC_1}\ket{\vec{p}}\bra{\vec{q}}_{\SC_1}=\sum_{\vec{p}\vec{q}\in\SC_1}\ket{\vec{p}}\bra{\vec{p}}_{\SC_1}=\id_{\SC_1}$ (here $p_k=q_k$ from Eq.~\eqref{eq:pq-1-n3}).

The latter allows us to get rid of all qubit indexes in $\SC_1$. Here, the remaining sets of equations
\begin{align}
    &p_k = q_k =0,  \quad \text{ for all qubits with index in $\SC_4$} \,.\label{eq:pq-n-n1}\\
    &p_k = q_k = z^{x}_k,   \quad \text{ for all qubits with index in $\SC_2$ and $\SC_3$}\,,  \label{eq:pq-1-n2}
\end{align}
are exactly like those in Eqs.~\eqref{eq:pqk-1} and  \eqref{eq:pqk-n} of Section~\ref{sec:prooftheo2}. One can now follow steps similar to those used in deriving Eq.~\eqref{eq:new-m1} and Eq.~\eqref{eq:srxj-thm2} to find that for a global cost function
\begin{align}
     \langle (\partial_\nu C_x)^2 \rangle  \leq g(n) \in \mathcal{O}(1/2^{2n}).
\end{align}

A similar analysis can be done for cross terms in $x$ to show that if $\partial_\nu C = (i/2N)\sum_{x=1}^N \partial_\nu C_x$, then 
\begin{align}
    \left\langle (\partial_\nu C)^2 \right\rangle \leq g(n) \in \mathcal{O}\left(\frac{1}{2^{2n}}\right). 
\end{align} 

Similarly, one con generalize the proof for local cost functions with no hidden layers, to the case when the DQNN has $L$ hidden layers.

\color{black}

\end{document}